\documentclass{article}
\usepackage{indentfirst} 

\usepackage{arxiv}
\usepackage{geometry}
\usepackage{hyperref}       % hyperlinks
\usepackage{times}              % times, fourier, fouriernc, pxfonts, mathpple
\usepackage{graphicx}
\usepackage{natbib}
\usepackage{amsmath,amsthm,amssymb,bm}

\renewcommand{\normalsize}{\fontsize{11}{11}\selectfont}

\title{Generalized Autoregressive Moving Average Models with GARCH Errors}

\date{This version: May, 2021}	% Here you can change the date presented in the paper title
\author{ {Tingguo Zheng} \\
	\small Department of Statistics and Data Science, School of Economics, Xiamen University, Xiamen, Fujian Province, China\\
	\small Wang Yanan Institute for Studies in Economics, Xiamen University, Xiamen, Fujian Province, China\\
	\small Email: \texttt{zhengtg@gmail.com} \\
	\And
	{Han Xiao} \\
	\small Department of Statistics, Rutgers University, Piscataway, NJ 08854, USA\\
	\small Email: \texttt{hxiao@stat.rutgers.edu} \\
    \And
	{Rong Chen} \\
	\small Department of Statistics, Rutgers University, Piscataway, NJ 08854, USA\\
	\small Email: \texttt{rongchen@stat.rutgers.edu} \\
}

\renewcommand{\shorttitle}{\textit{GARMA-GARCH}}
\newcommand{\F}{\mathcal{F}}
\newcommand{\Var}{\mathrm{Var}}
\newcommand{\logit}{\mathrm{logit}}

\newcommand{\hA}{{\boldsymbol{A}}}
\newcommand{\hB}{{\boldsymbol{B}}}

\newcommand{\vX}{{\mathfrak{X}}}
\newcommand{\sB}{\mathcal{B}}
\newcommand{\by}{\mathbf{y}}
\newcommand{\bfvarepsilon}{\boldsymbol{\varepsilon}}
\newcommand{\bfsigma}{\boldsymbol{\sigma}}

\theoremstyle{plain}

  \newtheorem{theorem}{Theorem}%[section]
  \newtheorem{lemma}{Lemma}%[section]

\theoremstyle{definition}
  \newtheorem{rem}{Remark}

\hypersetup{
pdftitle={A template for the arxiv style},
pdfsubject={q-bio.NC, q-bio.QM},
pdfauthor={David S.~Hippocampus, Elias D.~Striatum},
pdfkeywords={First keyword, Second keyword, More},
}

\begin{document}

\maketitle

\begin{abstract}
One of the important and widely used classes of models for non-Gaussian time series is the
generalized autoregressive model average models (GARMA), which specifies an ARMA structure for the
conditional mean process of the underlying time series. However, in many applications one often
encounters conditional heteroskedasticity. In this paper we propose a new class of models, referred to as GARMA-GARCH models, that jointly specify both the conditional mean and conditional variance
processes of a general non-Gaussian time series. Under the general modeling framework, we propose
three specific models, as examples, for proportional time series, nonnegative time series, and skewed and heavy-tailed financial time series. Maximum likelihood estimator (MLE) and quasi Gaussian MLE (GMLE) are used to estimate the parameters. Simulation studies and three applications
are used to demonstrate the properties of the models and the estimation procedures.
\end{abstract}

% keywords can be removed

\newpage

\section{Introduction}

The traditional autoregressive and moving average (ARMA) models for time series analysis assume the conditional mean of $y_t$ given the past information depends linearly on the past observations and past innovations. To capture empirical characteristics such as serial dependence between squared innovations, volatility clustering and other heteroskedasticity in many time series encountered in the real data, especially in economic and financial applications, additional assumptions on the conditional variance are often introduced and included in the model, resulting in models such as the generalized autoregressive conditional heteroskedasticity (GARCH) models proposed by \cite{Engle1982} and \cite{Bollerslev1986} and stochastic volatility (SV) models proposed by \cite{Taylor1982,Taylor1986}.  The resulting combined ARMA and GARCH processes are often referred to as the ARMA-GARCH models. %Practically, data transformation and detrending may be pretreated before model building.
In dealing with non-Gaussian time series exhibiting heavy tailed and asymmetric behaviors, the innovation (error) process in the ARMA-GARCH models is often assumed to follow certain non-Gaussian distributions, including Student-$t$ \citep{Bollerslev1987}, generalized error distribution \citep{Nelson1991}, skewed Student-$t$ \citep{LL2001}, and others. These approaches are often referred to as the
{\it innovations-based} approaches for dealing with non-Gaussian time series.

A different approach for modeling non-Gaussian time series is the {\it data-based} approach, in which one first assumes a non-Gaussian conditional distribution of $y_t$ given the past, then models the temporal evolution of the parameters of the distribution. Such an approach is often used to model time series of count, time series of positive random variables and proportional time series, including the autoregressive conditional duration models \citep{ER1998}, multiplicative error models \citep{Engle2002,EG2006}, Poisson and negative binomial models for discrete-valued or count data \citep{DDS2003,FLO2006,DW2009,FRT2009,FF2010,QLZ2020}, Beta autoregressive moving average (ARMA) models (Rocha and Cribari-Neto, 2009; Scher, Cribari-Neto, Pumi and Bayer, 2020), and many others.
A general class of such models is the Generalized ARMA (GARMA) model of \cite{BRS2003} and its
martingalized version (M-GARMA) model of \cite{ZXC2015}. Through a link function, GARMA and M-GARMA assumes an ARMA form for the conditional mean process.

Similar to the need of extending the approach of modeling only the conditional mean to jointly modeling the conditional mean and variance in the innovation-based approach mentioned above, there is also a need to extend the GARMA models to include the modeling of the conditional variance process, in order to address conditional heteroskedasticity often observed in applications. This is actually more important in the data-based approaches since here we directly model the conditional distribution. If the conditional distribution involves more than one parameter, then the conditional mean itself is often not sufficient to capture the time varying behavior of the conditional distribution. For example, if the distribution is parametrized by its mean and some other parameters, the GARMA model would need to require all the parameters except the mean to be fixed, and only the mean parameter to be time varying and carry the past information. Such an assumption limits the modeling flexibility for real applications.

In this paper we propose a data-based approach for modeling non-Gaussian time series by jointly modeling the conditional mean and variance processes, through a link function. It is an extension of the GARMA model, with an additional GARCH structure for the conditional variance process.

It is not straightforward to include a GARCH structure under the GARMA model framework of \cite{BRS2003}. The reason is that if the specified link function $h(\cdot)$ in the GARMA formation is not identity, the induced error sequence under the link function $\{h(y_t)-E[h(y_t)\mid \F_{t-1}]\}$ is not a martingale difference sequence (MDS) and thus its squared counterparts as a measure of the conditional variance are complicate and difficult to interpret. On the other hand, the Martingalized-GARMA (M-GARMA) model proposed by \cite{ZXC2015} provides a promising framework, allowing the error sequence being an MDS.

Based on the preceding discussions, this paper formally proposes the GARMA-GARCH model under the M-GARMA framework of \cite{ZXC2015} to capture the conditional heteroskedasticity of non-Gaussian time series.
We also introduce three specific models under the general framework: a log-Gamma-GARMA-GARCH model for non-negative time series, a logit-Beta-GARMA-GARCH model for proportional time series, and a generalized hyperbolic skew Student-$t$ (GHSST) distribution based GARMA-GARCH model for financial time series. We then present the maximum likelihood estimator (MLE) and quasi Gaussian MLE (GMLE) for estimating the parameters in the GARMA-GARCH models.

The rest of this paper is organized as follows. Section 2 introduces the GARMA-GARCH model, together with the three specific models and their properties. In Section 3, we introduce the MLE and GMLE for GARMA-GARCH models. Second 4 presents some simulation studies to demonstrate the finite sample performance of the estimators. In Section 5, empirical applications are carried out for realized volatility, U.S. personal saving rate and stock returns by using the three specific models, respectively. The last section concludes.

\section{The Model}

\subsection{GARMA-GARCH model}

The GARMA-GARCH model relies on a given parametric family of distributions $f(y\mid\gamma,\varphi)$ with parameters $\gamma$ and $\varphi$, and a y-link function $h(\cdot)$. Here both $\gamma$ and $\varphi$ can be vectors. Define the functions $g_\varphi(\gamma):=E[h(Y)]$ and $V_\varphi(\gamma):=\Var[h(Y)]$, where $Y$ follows the distribution $f(y\mid\gamma,\varphi)$. The GARMA-GARCH model assumes that the conditional distribution $p(y_t\mid\F_{t-1})$ is
\begin{equation}
    \label{mgarmagarch-dist}
    p(y_t\mid\F_{t-1})=f(y_t\mid\gamma_{t},\varphi),
\end{equation}
where $\gamma_{t}$ is determined by $\F_{t-1}$ through the conditional expectation and variance of $h(y_t)$, i.e.
\begin{align}
    \label{mgarmagarch-mean}
    g_{\varphi}(\gamma_t)=E[h(y_t)\mid \F_{t-1}]&=\phi_0+\sum_{j=1}^p\phi_jh(y_{t-j})+\sum_{j=1}^q\delta_j\varepsilon_{t-j},\\
    \label{mgarmagarch-variance}
    V_\varphi(\gamma_t)=\Var[h(y_t)\mid \F_{t-1}]&=\omega+\sum_{i=1}^r\alpha_i\varepsilon_{t-i}^2+\sum_{j=1}^s\beta_j\sigma_{t-j}^2,
\end{align}
where $\varepsilon_t:=h(y_t)-g_\varphi(\gamma_t)$, and $\sigma_t^2:=V_\varphi(\gamma_t)$. We  assume that $\varphi$ is a vector of time invariant parameters, and $\gamma_t$ can be uniquely solved from \eqref{mgarmagarch-mean} and \eqref{mgarmagarch-variance}, once $\F_{t-1}$ is given. Note that in most cases, this would require that $\gamma_t$ is a 2-dimensional vector.

We use $\theta$ to denote the set of all model parameters
$$\{\varphi,\phi_0,\phi_1,\ldots,\phi_p,\delta_1,\ldots,\delta_q,\omega,\alpha_1,\ldots,\alpha_r,\beta_1,\ldots,\beta_s\}.$$
Since \eqref{mgarmagarch-variance} is regarding the conditional variance, we impose the natural requirement that $\omega$ is strictly positive, and $\alpha_i$ ($i=1,\ldots,r$) and $\beta_j$ ($j=1,\ldots,s$) are nonnegative.

Adding $\varepsilon_t=h(y_t)-g_\varphi(\gamma_t)$ on both sides of \eqref{mgarmagarch-mean} leads to the following ARMA representation of $\{h(y_t)\}$.
\begin{equation}
    \label{eq:arma}
    h(y_t)=\phi_0+\sum_{j=1}^p\phi_jh(y_{t-j})+\varepsilon_t+\sum_{j=1}^q\delta_j\varepsilon_{t-j}.
\end{equation}
The joint model \eqref{eq:arma} and \eqref{mgarmagarch-variance} is in the standard ARMA-GARCH model form, for the transformed time series $h(y_t)$, except that the innovation process $\varepsilon_t/\sigma_t$ here is a martingale difference process, instead of an i.i.d sequence.

Similar to the GARCH model, the equation \eqref{mgarmagarch-variance} can also be represented in an ARMA form. For this purpose we define $\zeta_t:=\varepsilon_t^2-\sigma_t^2$, and add $\zeta_t$ on both sides of \eqref{mgarmagarch-variance}, resulting in the following equivalent representation:
\begin{equation}
    \label{eq:garch_arma}
    \varepsilon_t^2 = \omega+\sum_{i=1}^{r\vee s}(\alpha_i+\beta_i)\varepsilon_{t-i}^2 + \zeta_t - \sum_{j=1}^s\beta_j\zeta_{t-j},
\end{equation}
where $r\vee s$ denotes the maximum of $r$ and $s$.

Some remarks on different issues of the model are in order:

\begin{rem}[\it Relationship to the M-GARMA Model] The GARMA-GARCH model is a direct extension of the M-GARMA model of \cite{ZXC2015}, with the additional equation \eqref{mgarmagarch-variance} for the conditional variance of $h(y_t)$. The M-GARMA model only allows a one dimensional time varying parameter $\gamma_t$ in \eqref{mgarmagarch-dist}. The GARMA-GARCH model allows two parameters to be time varying and to depend on past information. It also closely mimics the standard ARMA-GARCH model widely used in applications.
\end{rem}

\begin{rem}[\it Relationship to the innovation-based ARMA-GARCH] The innovation-based ARMA-GARCH model assumes the form \eqref{eq:arma} and \eqref{mgarmagarch-variance} for the transformed series $h(y_t)$ with the error term $\varepsilon_t=\sigma_t e_t$, where $e_t$ are independent,  following a common distribution such as normal, Student's $t$, and skewed $t$. In this case, the GARCH process $\{\varepsilon_t\}$ is said to be strong, following definition of \cite{FZ2010} and \cite{DN1993}.
For such models, the conditional distribution of $y_t$ given $\F_{t-1}$
can be found through the inverse transformation of the transformed random variable $h(y_t)$.
In certain non-Gaussian time series such as time series of counts, the data-based approach taken here is more natural and the model is easier to interpret. Under the GARMA-GARCH model, the ``innovation" process $\{\varepsilon_t/\sigma_t\}$ is not an i.i.d. sequence, but a MDS, and the GARCH component is said to be semi-strong. The parallel between the two models allow us to obtain a quasi likelihood estimator for the GARMA-GARCH model by assuming that $h(y_t)$ follows a ARMA-GARCH model with i.i.d. Gaussian innovations.
\end{rem}

\begin{rem}[\it Mean and variance link functions]
Define
\begin{align}
   \label{mreg-linkfunc}
    \mu_t:& =g_{\varphi}(\gamma_t)=E[h(y_t)\mid\F_{t-1}],\\
        \label{mreg-varfunc}
    \sigma_t^2: & = V_\varphi(\gamma_{t})=\Var(\varepsilon_t\mid\F_{t-1})=\Var[h(y_t)\mid\F_{t-1}],
\end{align}
These functions link the
conditional mean and variance of $h(y_t)$ %or $\varepsilon_t$
given $\F_{t-1}$ to the underlying time varying parameter $\gamma_t$.
Different to the M-GARMA model, the link function $\mu_t$ here is not necessarily a function of conditional mean of $y_t$.
If the y-link function is identity, the mean and variance link functions are indeed directly
linked to the conditional mean and variance of $y_t$. Otherwise, they differ greatly.
The link functions of course depend on the form of parametrization of the conditional distribution
in \eqref{mgarmagarch-dist}.
\end{rem}

\begin{rem}[\it Solving time-varying parameters]
It is important to be able to solve the time-varying parameters $\gamma_t$ using \eqref{mgarmagarch-mean} and \eqref{mgarmagarch-variance}, given past information. This is because they are needed to evaluate the likelihood function of the model and so to carry out the maximum likelihood estimation. That is, we need to solve $\gamma_t=(\gamma_{1t},\gamma_{2t})$ based on given $\mu_t$ and $\sigma_t^2$ using the link functions
\eqref{mreg-linkfunc} and \eqref{mreg-varfunc} as a system of equations,
via exact or numerical methods. In practice, when there are more than one solution for the system, we can select the solution that maximizes the corresponding likelihood function value of the GARMA-GARCH model.
\end{rem}

\begin{rem}[\it The choice of the y-link function] One of the key components of a good GARMA-GARCH model is a properly chosen y-link function $h(\cdot)$, which specifies the explicit links between $\gamma_t$ and $(\mu_t, \sigma^2_t)$. The y-link function is similar to the link function of the generalized linear models of \cite{MN1989}. It is a model assumption that is based on the problem at hand, and can be checked and sometimes justified with sensible model validation statistics and procedures.
\end{rem}

\begin{rem}[\it Dimension of the time-varying parameters]
In general, the dimension of the time-varying parameter $\gamma_t$ is assumed to be two under the GARMA-GARCH model, as discussed earlier.
If the dimension of $\gamma_t$ is one, one may choose to use either the conditional mean relationship \eqref{mgarmagarch-mean} or the conditional variance relationship \eqref{mgarmagarch-variance}, but not both. An alternative is to adopt a generalized moment method approach by solving the one dimensional $\gamma_t$ with two moment equations. This is an interesting problem to be further investigated.
When the dimension of $\gamma_t$ is larger than 2, two equations with the mean and variance link functions are not enough to determine the time-varying parameters. In this case, we may extend the GARMA-GARCH model to include modeling assumptions of the higher moments of $y_t$ given $\F_{t-1}$.
\end{rem}

\begin{rem}[\it Extensions to multivariate non-Gaussian time series]
In this paper we study univariate time series $y_t$ through GARMA-GARCH model. The model can be easily extended to multivariate time series, with a multivariate y-link function $h(\cdot)$, and a vector ARMA form in place of \eqref{mgarmagarch-mean} and a multivariate GARCH form in place of
\eqref{mgarmagarch-variance}. To avoid various ambiguities in multivariate ARMA, vector AR models can be used \citep{Lu2005,Tsay2014}. There are also many different multivariate GARCH models, including those in \cite{AHL2020,Engle2002b,EK1995}.
The dependencies among the components of $\mu_t$ (as a vector) and $\sigma^2_t$ (as a covariance matrix) induce dependencies among the components of $y_t$.
\end{rem}

\subsection{Some specific models}\label{GARMA-GARCH}

{In this section we introduce three specific models under the general GARMA-GARCH model framework, with specific choices of the conditional distribution $f(\cdot)$, and the y-link function $h(\cdot)$. They are designed for certain types of non-Gaussian time series.}

To simplify the model formula, define the following characteristic polynomials: $\phi(z)=1-\phi_1 z-\cdots-\phi_p z^p$, $\delta(z)=1+\delta_1 z+\cdots+\delta_q z^q$, $\alpha(z)=\alpha_1 z+\cdots+\alpha_r z^r$, and $\beta(z)=1-\beta_1 z-\cdots-\beta_s z^s$. Let $L$ be the backward shift operator such that $Ly_t=y_{t-1}$.

\vspace{1em}

\noindent
{\bf (1) Log-Gamma-GARMA-GARCH model: } Suppose $y_t$ is a non-negative continuous random variable. Using Gamma distribution as the conditional distribution and log function as the y-link function, the log-Gamma-GARMA($p,q$)-GARCH($r,s$) model for a time series $\{y_t\}$ is given by
\begin{align}\label{gammadist}
    y_t\mid\F_{t-1}&\sim\mathrm{Gam}(c_t,c_t/\eta_t),\quad
    \phi(L)\log y_t=\phi_0+\delta(L)\varepsilon_t,\quad
    \beta(L)\sigma_t^2=\omega+\alpha(L)\varepsilon_t^2,
\end{align}
with $\varepsilon_t=\log y_t-\mu_t=\log y_t-g(\eta_t,c_t)$, where $c_t$ and $c_t/\eta_t$ are the shape and rate parameters of Gamma distribution respectively, $\eta_t$ and $\mu_t$ are the conditional expectations of $y_t$ and $\log y_t$ respectively, and $\sigma_t^2$ is the conditional variance of $\varepsilon_t$ or $h(y_t)$. Moreover, the resulting {conditional mean} %link function
and variance link functions are expressed as
\begin{align}
    \label{gamma-linkfunc}
    g(\eta_t,c_t)&=\mu_t=\log\eta_t+\psi(c_t)-\log c_t,\\
    \label{gamma-varfunc}
    V(\eta_t,c_t)&=\sigma_t^2=\psi_1(c_t),
\end{align}
where $\psi(\cdot)$ is the digamma function, and $\psi_1(\cdot)$ is the trigamma function.

In this model \eqref{gammadist}, we have no fixed parameter $\varphi$, and $(c_t,\eta_t)$ are the time-varying parameters of the distribution. Given the values of $\mu_t$ and $\sigma_t^2$, we can first obtain the unique root $c_t=\psi_1^{-1}(\sigma_t^2)$ from \eqref{gamma-varfunc} via the numerical methods or the bisection method, and then substituting it into the link function \eqref{gamma-linkfunc} to calculate $\eta_t=\exp[\mu_t+\log c_t-\psi(c_t)]$.

Without the conditional variance process, \cite{ZXC2015} proposed the log-Gamma-M-GARMA$(p,q)$
model in which the conditional distribution is assumed to be $\mathrm{Gam}(c,c/\eta_t)$ with a time invariant parameter $c$, the y-link function used is the log function, and the conditional mean process assumes an ARMA form. Comparing to the log-Gamma-GARMA-GARCH model, its shape parameter $c$ is time-invariant, and it does not have all the GARCH parameters in the conditional variance process.

\begin{rem}\label{rk:gamma0}
Here we point out a special feature of the log-Gamma-M-GARMA process. Because the y-link function is $\log(\cdot)$, the conditional distribution of $\varepsilon_t$ is actually completely determined by the shape parameter $c_t$, no matter what value $\eta_t$ takes. To see this, note that $\tilde y_t:=(c_t/\eta_t)y_t\sim\mathrm{Gam}(c_t,1)$, and $$\varepsilon_t=\log y_t-g(\eta_t,c_t) = y_t - \psi(c_t) - \log\eta_t + \log c_t = \log\tilde y_t -\psi(c_t).$$
On the other hand, according to \eqref{gamma-varfunc}, $c_t$ is uniquely determined by $\sigma_t^2$ (note that $\psi_1(\cdot)$ is a strictly decreasing function on the positive real line). Consequently, we can give an equivalent definition of the log-Gamma-GARMA-GARCH model. We first define the process $\{\varepsilon_t\}$ according to
\begin{align}\label{gammadist0}
    \tilde y_t\mid\F_{t-1}&\sim\mathrm{Gam}(c_t,1),\quad
    \varepsilon_t=\tilde y_t-\psi(c_t),\quad
    \beta(L)\sigma_t^2=\omega+\alpha(L)\varepsilon_t^2.
\end{align}
Once $\{\varepsilon_t\}$ is given, the $\{h(y_t)\}$ process can be constructed
by \eqref{eq:arma} using the pre-generated $\{\varepsilon_t\}$ process. This equivalent definition reveals that $\{\varepsilon_t\}$ can be defined without referring to the conditional mean equation \eqref{mgarmagarch-mean}, and $\{h(y_t)\}$
is simply a linear ARMA process with innovations $\{\varepsilon_t\}$. This feature simplifies the investigation of various probabilistic properties of the process.
\end{rem}

\vspace{1em}

\noindent
{\bf (2) Logit-Beta-GARMA-GARCH model: } This model is based on the logit-Beta-M-GARMA model proposed by \cite{ZXC2015}. Suppose a proportional time series $\{y_t\}$ lies in interval (0,1). Using {the Beta distribution as the conditional distribution} and the logit transformation as the y-link function, i.e., $h(y_t)=\logit(y_t)=\log[y_t/(1-y_t)]$, the logit-Beta-GARMA($p,q$)-GARCH($r,s$) model is given by
\begin{align}\label{betadist}
    y_t\mid\F_{t-1}&\sim\mathrm{Beta}(a_t,b_t),\quad
    \phi(L)\logit(y_t)=\phi_0+\delta(L)\varepsilon_t,\quad
    \beta(L)\sigma_t^2=\omega+\alpha(L)\varepsilon_t^2,
\end{align}
with $\varepsilon_t=\logit(y_t)-\mu_t=\logit(y_t)-g(a_t,b_t)$, where $a_t$ and $b_t$ are two positive parameters of the Beta distribution, $\mu_t$ is the conditional expectation of $\logit(y_t)$, and $\sigma_t^2$ is the conditional variance of $\varepsilon_t$ or $\logit(y_t)$.
The resulting mean and variance functions are expressed as
\begin{align}
    \label{beta-linkfunc}
    g(a_t,b_t)&=\mu_t=\psi(a_t)-\psi(b_t),\\
    \label{beta-varfunc}
    V(a_t,b_t)&=\sigma_t^2=\psi_1(a_t)+\psi_1(b_t),
\end{align}
where $\psi(\cdot)$ and $\psi_1(\cdot)$ are the digamma and trigamma functions.

In this case, we also have no fixed parameter $\varphi$, but two time-varying parameters of the distribution, $a_t$ and $b_t$. Although this system of equations \eqref{beta-linkfunc} and \eqref{beta-varfunc} is nonlinear, the solution given the values of $\mu_t$ and $\sigma_t^2$ is unique. Since both digamma and trigamma are strictly monotonic on $(0,\infty)$,
there must exist a unique solution such that $a_t=\psi^{-1}[\mu_t+\psi(b_t)]$ (for a given $b_t$), and also a unique solution of $b_t$ to the nonlinear equation $\sigma_t^2=\psi_1\{\psi^{-1}[\mu_t+\psi(b_t)]\}+\psi_1(b_t)$
since $\psi(\cdot)$ is monotonically increasing and $\psi_1(\cdot)$ is monotonically decreasing.
In practice, since the system of \eqref{beta-linkfunc} and \eqref{beta-varfunc} is highly nonlinear, numerical methods are used to solve both $a_t$ and $b_t$ simultaneously.

Again, without the conditional variance process, \cite{ZXC2015} proposed the logit-Beta-M-GARMA$(p,q)$ model in which the conditional distribution is assumed to be $\mathrm{Beta}(\tau\eta_t,(1-\tau)\eta_t)$ with a time invariant parameter $\tau$, the y-link function used is the logit function, and the conditional mean process assumes an ARMA form. Comparing to the logit-Beta-GARMA-GARCH model, it has one extra time-invariant parameter $\tau=a_t+b_t$ (with $a_t=\tau\eta_t, b_t=(1-\tau)\eta_t$), but does not have all the GARCH parameters in the conditional variance process.

\vspace{1em}

\noindent
{\bf (3) GHSST-GARMA-GARCH model: } To model leptokurtic and skewed financial time series $\{y_t\}$, we consider a generalized hyperbolic skew Student-t (GHSST) distribution proposed by \cite{AH2006}. It is shown that the GHSST distribution can exhibit unequal thickness in both tails, contrary to other skewed extensions of the Student-$t$, and %argue
that this offers more modeling flexibility. Using the identity y-link function, we propose the following GHSST-GARMA($p,q$)-GARCH($r,s$) model
\begin{align}\label{GHSST-GARCH}
    y_t\mid\F_{t-1}&\sim\mathrm{GHSST}(\xi_t,\varsigma_t,\nu,\tau),\quad
    \phi(L)y_t=\phi_0+\delta(L)\varepsilon_t,\quad
    \beta(L)\sigma_t^2=\omega+\alpha(L)\varepsilon_t^2,
\end{align}
where $\varepsilon_t=y_t-\mu_t$.
The conditional density of $y_t$ follows the GHSST distribution
\begin{align}\label{GHSSTdens}
    f(y_t\mid\xi_t,\varsigma_t,\nu,\tau)&=\frac{2^{\frac{1-\nu}{2}}\varsigma_t^\nu|\tau|^{\frac{\nu+1}{2}}K_{\frac{\nu+1}{2}}\left(\sqrt{\tau^2(\varsigma_t^2+(y_t-\xi_t)^2)}\right)\exp(\tau(y_t-\xi_t))}{\Gamma(\frac{\nu}{2})\sqrt{\pi}\left(\sqrt{\varsigma_t^2+(y_t-\xi_t)^2}\right)^{\frac{\nu+1}{2}}},
\end{align}
where $\xi_t$ and $\varsigma_t$ $(>0)$ are the time varying location and scale parameters of the GHSST distribution, and $\nu$ and $\tau$ are the degrees of freedom and shape parameters. The function $K_j(x)=\frac{1}{2}\int_0^\infty z^{j-1}e^{-\frac{x}{2}(z+z^{-1})}dz$ for $x>0$ denotes the modified Bessel function of the third kind and of order $j\in\mathbb{R}$ \citep{AS1972}. The conditional density $f(y_t\mid\xi_t,\varsigma_t,\nu,\tau)$ can be recognized as the density of a non-central (scaled) Student's $t$-distribution with $\nu$ degrees of freedom. In particular when $\tau=0$, it becomes the standard Student's $t$-distribution.

In model \eqref{GHSST-GARCH},
the corresponding conditional mean and variance functions are given by
\begin{align}
    \label{linkfunc-skew-t}
    g_{\nu,\tau}(\xi_t,\varsigma_t)&=\mu_t=\xi_t+\frac{\tau\varsigma_t^2}{\nu-2},\\
    \label{varfunc-skew-t}
    V_{\nu,\tau}(\xi_t,\varsigma_t)&=\sigma_t^2=\frac{\varsigma_t^2}{\nu-2}+\frac{2\tau^2\varsigma_t^4}{(\nu-2)^2(\nu-4)}.
\end{align}
The variance is only finite when $\nu>4$, as opposed to the symmetric Student's $t$-distribution which only requires $\nu>2$.
In this case, we can first solve the roots of the polynomial on the right hand side of \eqref{varfunc-skew-t} and select the unique positive real root as the value of $\varsigma_t$, that is,
\[
    \varsigma_t=\sqrt{\frac{-b_0+b_0\sqrt{1+8\tau^2\sigma_t^2/(\nu-4)}}{4\tau^2}},
\]
where $b_0=(\nu-2)(\nu-4)$. Then, the other time-varying parameter $\xi_t$ can be calculated as $\xi_t=\mu_t-\tau\varsigma_t^2/(\nu-2)$. %given $\varsigma_t$.

Under the model framework presented in this study, we can estimate all parameters in \eqref{GHSST-GARCH} using the maximum likelihood estimation procedure directly. We note that \cite{Deschamps2012} proposed another ARMA-GARCH form by assuming $\varepsilon_t$ as a mixture of normal and inverted Gamma random variables, the likelihood function cannot be expressed analytically and thus the MLE estimator is infeasible. In practice under the mixture representation, one has to use the Markov Chain Monte Carlo methods for the parameter estimation.

\subsection{Stationarity and Ergodicity}

To study the stationarity and ergodicity of the GARMA-GARCH model, we consider the state space representation of the model and apply the theory of Markov Chains on a general state space. The theoretical tools involved in our analysis are covered by the classical treatise in \cite{meyn:2009}, especially Chapter~15. We will not repeat the concepts and terminologies about the Markov Chains here, but refer the readers to the aforementioned book.

To represent the GARMA-GARCH model in the state space form, we start from \eqref{eq:arma} and \eqref{eq:garch_arma}. Without loss of generality,
assume for \eqref{eq:arma}, $q=p-1$, and at least one of $\phi_p$ and $\delta_{p-1}$ is
nonzero, and similarly for \eqref{eq:garch_arma}, $s=r-1$, and at least one of $\alpha_r$ and $\beta_{r-1}$ is
nonzero.
Define the square matrices
\begin{equation}
  \label{eq:garmat}
  \Phi=
  \begin{pmatrix}
    \phi_1 & \phi_2 & \cdots & \phi_{p-1} & \phi_p \\
    1 & 0 & \cdots & 0 & 0\\
    0 & 1 & \cdots & 0 & 0\\
    \vdots & \vdots & \ddots & \vdots & \vdots\\
    0 & 0 & \cdots & 1 & 0
  \end{pmatrix},\quad
%   \Phi_1=
%   \begin{pmatrix}
%     \phi_1+\delta_1 & \phi_2+\delta_2 & \cdots & \phi_{p-1}+\delta_{p-1} & \phi_p \\
%     0 & 0 & \cdots & 0 & 0\\
%     0 & 0 & \cdots & 0 & 0\\
%     \vdots & \vdots & \ddots & \vdots & \vdots\\
%     0 & 0 & \cdots & 0 & 0
%   \end{pmatrix},
\bm A =  \begin{pmatrix}
    \alpha_1+\beta_1 & \alpha_2+\beta_2 & \cdots & \alpha_{r-1} + \beta_{r-1} & \alpha_r \\
    1 & 0 & \cdots & 0 & 0\\
    0 & 1 & \cdots & 0 & 0\\
    \vdots & \vdots & \ddots & \vdots & \vdots\\
    0 & 0 & \cdots & 1 & 0
  \end{pmatrix},
\end{equation}
and set $\delta=(1,\delta_1,\ldots,\delta_{p-1})'$ and $\beta=(1,-\beta_1,\ldots,-\beta_{r-1})'$. Let $\mu=\phi_0/(1-\sum_i\phi_i)$ and $\sigma^2=w/(1-\sum_i\alpha_i-\sum_j\beta_j)$.
We first define a Markov Chain $\{(X_t',Z_t')'\}$ on $\mathbb{R}^{p+r}$, where $X_t$ and $Z_t$ are $p$- and $r$-dimensional, respectively.
Given $X_{t-1}$ and $Z_{t-1}$, we generate $y_t$ as
\begin{equation}
  \label{eq:MC1y}
  y_t\sim f(y_t\mid\gamma_t,\varphi),
\end{equation}
where the parameter $\gamma_t$ is determined by
\begin{equation*}
    g_\varphi(\gamma_t)=\mu_t:=\mu+\delta'\Phi X_{t-1} \quad\hbox{and}\quad V_\varphi(\gamma_t)=\sigma_t^2:=\sigma^2+\beta'\bm A Z_{t-1}.
\end{equation*}
We then set $\varepsilon_t=h(y_t)-\mu_t$, $\zeta_t=\varepsilon_t^2-\sigma_t^2$, and define
\begin{equation}
  \label{eq:MC1}
  \begin{pmatrix}
  X_t \\
  Z_t
  \end{pmatrix}
  = \begin{pmatrix}
  \Phi & \bm 0 \\
  \bm 0 & \bm A
  \end{pmatrix}
  \begin{pmatrix}
  X_{t-1} \\
  Z_{t-1}
  \end{pmatrix}+
(\varepsilon_t,0,\ldots,0,\zeta_t,0,\ldots,0)'.
\end{equation}
Clearly $\{(X_t',Z_t')'\}$ is a time-homogeneous Markov chain. It
can be shown that $h(y_t)= \delta'X_t$, and the $\{y_t\}$ process
defined in \eqref{eq:MC1y} satisfies the recursive relationship \eqref{mgarmagarch-dist}, \eqref{mgarmagarch-mean} and
\eqref{mgarmagarch-variance}.

In general, the stationarity and ergodicity of the Markov Chain $\{(X_t',Z_t')'\}$ can be implied by the {\it geometric drift condition} \citep[see Chapter~15 of][]{meyn:2009}. Such conditions depend on the conditional distribution assumption and the y-link function used. As an illustration, we next show when this condition is fulfilled for the log-Gamma-GARMA-GARCH and logit-Beta-GARMA-GARCH models.
Recall that $\phi(z)=1-\phi_1 z-\cdots -\phi_p z^p$, and $\hA$ is the $r\times r$ matrix defined in \eqref{eq:MC1}. Let $\hA_1$ be the $r\times r$ matrix whose first row is $(\alpha_1,\ldots,\alpha_r)$, and other elements are zero. Define the sequence of matrices $\{\hB_k\}$ recursively as $\hB_0=I$, and $\hB_k=\hA'\hB_{k-1}\hA+5\hA_1'\hB_{k-1}\hA_1$ for $k\geq 1$.

\begin{theorem}
\label{thm:gamma}
Consider the log-Gamma-GARMA-GARCH model \eqref{gammadist}. Assume (i) for some $h$, the operator norm of $\hB_h$ is strictly less than one; and (ii) $\phi(z)\neq 0$ for $|z|\leq 1$. Then the equation \eqref{gammadist} admits a solution $\{y_t\}$ that is strictly stationary, and satisfying $E[h(y_t)]^4<\infty$.
\end{theorem}

\begin{theorem}
\label{thm:beta}
Consider the logit-Beta-GARMA-GARCH model \eqref{betadist}.
If for some $h$, the operator norms of $\hB_h$ and $\Phi^h$ are both strictly less than one, then the equation \eqref{gammadist} admits a solution $\{y_t\}$ that is strictly stationary, and satisfying $E[h(y_t)]^4<\infty$.
\end{theorem}
Note that in Theorem~\ref{thm:beta}, the condition that the operator norm of $\Phi^h$ is strictly less than one for some $h$ entails that $\phi(z)\neq 0$ for $|z|\leq 1$.
The proofs of the theorems are in Appendix.

\newpage
\section{Parameter Estimation}\label{MLE}

Let $\theta$ be the fixed and unknown parameter vector containing all model parameters. We partition the parameter vector into three sub-vectors: $\theta_{\mathrm{arma}}$ includes all the parameters in the ARMA process \eqref{mgarmagarch-mean}, $\theta_{\mathrm{garch}}$ includes all the parameters in the GARCH process \eqref{mgarmagarch-variance}, and $\varphi$ includes all time invariant parameters in the conditional distribution \eqref{mgarmagarch-dist}. Specifically, $\theta=(\theta'_{\mathrm{arma}},\theta'_{\mathrm{garch}},\varphi')'$, where
\[
    \theta_{\mathrm{arma}}=(\phi_0,\phi_1,\ldots,\phi_p,\delta_1,\ldots,\delta_q), \quad \theta_{\mathrm{garch}}=(\omega,\alpha_1,\ldots,\alpha_r,\beta_1,\ldots,\beta_s).
\]
We further set the initial values $\F_0=\by_0=(y_0,\ldots,y_{1-m})$, $\bfvarepsilon_0=(\varepsilon_0,\varepsilon_{-1},\ldots,\varepsilon_{1-m})$, $\bfvarepsilon_0^2=(\varepsilon_0^2,\ldots,\varepsilon_{1-m}^2)$ and $\bfsigma_0^2=(\sigma_0^2,\ldots,\sigma_{1-m}^2)$, where $m=\max(p,q,r,s)$.

Suppose the available data set is $\{y_{1-m},\ldots,y_0,y_1,\ldots,y_T\}$. We next introduce two approaches for parameter estimation. The first is based on the likelihood, and the second is based on quasi Gaussian likelihood with additional conditional likelihood
estimation.

\subsection{Maximum likelihood estimation}

The log-likelihood function conditional on the initial values ($\by_0$, $\bfvarepsilon_0$, $\bfvarepsilon_0^2$, and $\bfsigma_0^2$) is
\begin{align}\label{loglik0}
    L_T(\theta)=\sum_{t=1}^T\ell_t(\theta)=\sum_{t=1}^T\log f(y_t \mid \gamma_t,\varphi),
\end{align}
where $\gamma_t$ for $t=1,\ldots,T$ can be solved by the system consisting of the link and variance functions given by \eqref{mreg-linkfunc} and \eqref{mreg-varfunc}.

Given a set of parameters $\theta$, $\mu_t=g_\varphi(\gamma_t)$ and $\sigma_t^2=V_\varphi(\gamma_t)$ can be recursively obtained for $t=1,\ldots,T$. Recursively solving the system consisting of mean and variance link functions  yields $\gamma_t$ for $t=1,\ldots, T$. By plugging $\gamma_t$ into (\ref{loglik0}), we can evaluate the log-likelihood function $L_t(\theta)$. The maximum likelihood estimate is then obtained by maximizing (\ref{loglik0}), using nonlinear optimization procedures.

In practice, one can set the initial values $\varepsilon_0=\ldots=\varepsilon_{1-m}$ to zero and $\varepsilon_0^2=\cdots=\varepsilon_{1-m}^2=\sigma_0^2=\ldots=\sigma_{1-m}^2$ to the unconditional variance to reduce the complexity, a common practice in time series estimation. The resulting estimate $\hat\theta$ is the conditional maximum likelihood estimate. One can also use the quasi Gaussian MLE for the ARMA and GARCH parameters and its corresponding ML estimate for $\varphi$ (see the next subsection) as the initial values for this procedure.

The theory of \cite{HH1980} can be applied to study the asymptotic distribution of the MLE. However, for concrete GARMA-GARCH models, sufficient conditions that ensure the asymptotic normality are still under investigation. Regardless of the technical issues, we provide a reasonable formula for the asymptotic covariance matrix. Let
\begin{equation*}
  u_t(\theta) = \frac{\partial \log f(y_t|\gamma_t,\varphi)}{\partial\theta}.
\end{equation*}
Let $\theta_0$ be the true parameter. Under regularity conditions, $\{u_t(\theta_0)\}$ is an MDS with respect to $\{\F_t\}$. Define
\begin{equation*}
  I_T(\theta_0) = \sum_{t=1}^T E_{\theta_0}[u_t(\theta_0)(u_t(\theta_0))'|\F_{t-1}].
\end{equation*}
Under regularity conditions, it holds that
\begin{equation*}
  [I_T(\theta_0)]^{-1/2}(\hat\theta-\theta_0) \Rightarrow N(0,I).
\end{equation*}

For some GARMA-GARCH models presented in this study, for a given $\theta_0$, the conditional information $E_{\theta_0}[u_t(\theta_0)(u_t(\theta_0))'|\F_{t-1}]$ can be calculated explicitly. By substituting the estimate $\hat\theta$, we get an estimate $I_T(\hat\theta)$ of the Fisher information. In practice, the standard errors of the estimates can also be obtained by evaluating the Hessian matrix of the log-likelihood function (\ref{loglik0}) at the MLE. Our empirical experiences have shown that this estimator works very well.

\subsection{Gaussian ML estimation}

Just based on \eqref{mgarmagarch-mean} and \eqref{mgarmagarch-variance}, one can use quasi Gaussian likelihood to estimate the ARMA and GARCH parameters quickly. This estimate will be referred to as quasi Gaussian maximum likelihood estimator (GMLE).
%In the literature it is also called an quasi-maximum likelihood estimator.
Let $\vartheta=(\theta'_{\mathrm{arma}},\theta'_{\mathrm{garch}})'$. The GMLE for the GARMA-GARCH model is equivalent to a minimizer of
\begin{align}
    Q_T(\vartheta)=\frac{1}{T}\sum_{t=1}^Tq_t=\frac{1}{T}\sum_{t=1}^T\left(\log\sigma_t^2(\vartheta)+\frac{\varepsilon_t^2(\vartheta)}{\sigma_t^2(\vartheta)}\right),
\end{align}
where $\varepsilon_t(\vartheta)$ and $\sigma_t^2(\vartheta)$ are defined by \eqref{mgarmagarch-mean} and \eqref{mgarmagarch-variance}. This objective function can be evaluated relatively cheaply, and many algorithms are available for finding its minima.

The asymptotic properties such as consistency and asymptotic normality (CAN) of the GMLE can be obtained based on the existing studies. For example many studies such as \cite{Lumsdaine1996}, \cite{BHK2003}, \cite{Ling2007} and \cite{FZ2004,FZ2010,FZ2012} have investigated the CAN of the GMLE when the GARCH process of $\varepsilon_t$ in our model is strong with i.i.d. innovations. \cite{LH1994} studied the CAN of the GMLE for strictly stationary semi-strong GARCH(1,1) model and \cite{Escanciano2009} further proved the CAN of the GMLE for general GARCH($p,q$) model under martingale assumptions.

The general CAN theory of the GARMA-GARCH model is very challenging to establish, which we leave to the future work. Here we only discuss how the result in \cite{Escanciano2009} can be adapted to give the asymptotic normality of the GMLE for the two examples given in Section~\ref{GARMA-GARCH}.

For the log-Gamma-GARMA-GARCH or the logit-Beta-GARMA-GARCH model, we assume the conditions of Theorem~\ref{thm:gamma} or Theorem~\ref{thm:beta} hold, respectively. Assume in addition that (i) the parameter space $\Theta$ of $\vartheta$ is compact and $\vartheta$ belongs to the interior of $\Theta$; (ii) $\phi(z)$ and $\delta(z)$ have no common roots, and $\delta(z)\neq 0$ when $|z|\leq 1$; (iii) $\alpha(z)$ and $\beta(z)$ have no common roots, $\beta(z)\neq 0$ when $|z|\leq 1$, and $\alpha(1)\neq 0$, $\alpha_r+\beta_s>0$; and (iv) $E|\varepsilon_t|^{4+\delta}<\infty$ for some $\delta>0$. Under these regularity conditions, it holds that
\begin{align}
    \sqrt{T}(\bar\vartheta-\vartheta_0)\Rightarrow N(0,I(\vartheta_0)^{-1}),
\end{align}
where $I(\vartheta_0)$ denotes the information matrix, which is defined as minus the expected Hessian, i.e., $I(\vartheta_0)=-E\left(\frac{\partial^2q_t(\vartheta)}{\partial\vartheta\partial\vartheta'}|_{\vartheta=\vartheta_0}\right)$

The GML estimation procedure does not provide an estimator for the invariant parameter
$\varphi$. However, given the GML estimate $\bar\vartheta$, we can obtain the estimated
ARMA residual $\hat{\varepsilon}_t$ and the estimated conditional variance $\hat{\sigma}^2_t$. Set $\hat{\mu}_t=h(y_t)-\hat{\varepsilon}_t$. Given $\varphi$, and with $\hat{\mu}_t$ and $\hat{\sigma}_t^2$, we can obtain $\hat{\gamma}_t(\varphi)$ by solving the system of equations $g_{\varphi}(\gamma_t)=\hat{\mu}_t$ and $V_\varphi(\gamma_t)=\hat{\sigma}_t^2$. We emphasize that
$\hat{\gamma}_t$ depends on $\varphi$. Now define the log-likelihood function of $\varphi$ given fixed
$\bar\vartheta$ as
\[
\bar L_T(\varphi)=\sum_{t=1}^T\log f(y_t|\hat{\gamma}_t(\varphi),\varphi).
\]
The maximizer of $L_T(\varphi)$ can be treated as the pseudo-ML estimator of $\varphi$.

\begin{rem}[Model identification and model checking]
Both AIC and BIC can be used for model/order selection. Standard time series analysis tools such as ACF, PACF and EACF \citep{TT1984,CMC2013} can be used to obtain the preliminary lag orders for the ARMA process. Empirical studies have shown that
$r=s=1$ are often sufficient for the GARCH process in practice and can be used as the starting point of the model building process. For model checking, the portmanteau Q-statistic of the residuals and the squared standardized residuals can be used, as a standard exercise for time series model checking, especially for time series with potential conditional heteroskedasticity, along with other residual analysis tools.
\end{rem}

\section{Simulation Examples}\label{simulation}

In this section, we investigate finite sample performances of the preceding two estimators (named as GMLE and MLE) under the log-Gamma-GARMA-GARCH and logit-Beta-GARMA-GARCH models introduced in Section \ref{GARMA-GARCH}. We also present empirical evidence regarding the impact of model mis-specification when ignoring the GARCH process in the model. Without the GARCH part, the models
become the log-Gamma-M-GARMA and logit-Beta-M-GARMA models in \cite{ZXC2015}.

All estimates are obtained through a constraint optimization procedure that uses the MaxSQP algorithm, implementing a sequential quadratic programming technique, see \cite{NW1999}. We also use the solver for systems of nonlinear equations (SolveNLE) in OxMetrics software (see Doornik, 2007) to recursively obtain the solution of $\gamma_t$.

\subsection{Log-Gamma-GARMA(1,1)-GARCH(1,1) model}

\begin{table}[!t]
\renewcommand{\arraystretch}{1.00}
\small \tabcolsep 3.7mm
\caption{Simulation results of the log-Gamma-GARMA-GARCH model}\label{table1}
\begin{center}
\begin{tabular}{ccccccccccccc}
\hline\hline
           &          &  \multicolumn{2}{c}{Log-Gamma-M-GARMA}  &&\multicolumn{2}{c}{Log-Gamma-GARMA-GARCH}\\ \cline{3-4}\cline{6-7}
 Parameter &     True &               GMLE &                MLE &&               GMLE &                MLE \\ \hline
           &                                  \multicolumn{6}{c}{$T=100$}                                \\
  $\phi_0$ &  ~~~0.00 & $-$0.0006 (0.1189) & ~~~0.0025 (0.0791) && ~~~0.0052 (0.0052) & ~~~0.0033 (0.0758) \\
  $\phi_1$ &  ~~~0.95 & ~~~0.8551 (0.1741) & ~~~0.8635 (0.1293) && ~~~0.8642 (0.1272) & ~~~0.8710 (0.1230) \\
$\delta_1$ &  $-$0.65 & $-$0.5768 (0.1941) & $-$0.5873 (0.1618) && $-$0.5850 (0.1640) & $-$0.5886 (0.1555) \\
       $c$ &          & ~~~2.8333 (0.8293) & ~~~2.8497 (0.8247) &&                 -- &                 -- \\
  $\omega$ &  ~~~0.02 &                 -- &                 -- && ~~~0.0882 (0.1044) & ~~~0.0830 (0.1045) \\
$\alpha_1$ &  ~~~0.06 &                 -- &                 -- && ~~~0.0712 (0.1002) & ~~~0.0638 (0.0848) \\
 $\beta_1$ &  ~~~0.90 &                 -- &                 -- && ~~~0.7227 (0.2927) & ~~~0.7401 (0.2913) \\\hline
           &                                  \multicolumn{6}{c}{$T=500$}                                \\
  $\phi_0$ &  ~~~0.00 & ~~~0.0007 (0.0150) & $-$0.0003 (0.0181) && ~~~0.0011 (0.0144) & ~~~0.0003 (0.0139) \\
  $\phi_1$ &  ~~~0.95 & ~~~0.9369 (0.0238) & ~~~0.9259 (0.0274) && ~~~0.9374 (0.0245) & ~~~0.9379 (0.0214) \\
$\delta_1$ &  $-$0.65 & $-$0.6396 (0.0538) & $-$0.6371 (0.0506) && $-$0.6395 (0.0503) & $-$0.6393 (0.0436) \\
       $c$ &          & ~~~2.5598 (0.3815) & ~~~2.5706 (0.3733) &&                 -- &                 -- \\
  $\omega$ &  ~~~0.02 &                 -- &                 -- && ~~~0.0417 (0.0606) & ~~~0.0376 (0.0561) \\
$\alpha_1$ &  ~~~0.06 &                 -- &                 -- && ~~~0.0655 (0.0357) & ~~~0.0623 (0.0310) \\
 $\beta_1$ &  ~~~0.90 &                 -- &                 -- && ~~~0.8458 (0.1538) & ~~~0.8574 (0.1424) \\\hline
           &                                  \multicolumn{6}{c}{$T=2000$}                                \\
  $\phi_0$ &  ~~~0.00 & ~~~0.0003 (0.0060) & ~~~0.0000 (0.0076) && ~~~0.0003 (0.0056) & ~~~0.0002 (0.0057) \\
  $\phi_1$ &  ~~~0.95 & ~~~0.9467 (0.0111) & ~~~0.9331 (0.0135) && ~~~0.9468 (0.0105) & ~~~0.9467 (0.0095) \\
$\delta_1$ &  $-$0.65 & $-$0.6473 (0.0255) & $-$0.6436 (0.0246) && $-$0.6471 (0.0105) & $-$0.6464 (0.0110) \\
       $c$ &          & ~~~2.4904 (0.2099) & ~~~2.4987 (0.2029) &&                 -- &                 -- \\
  $\omega$ &  ~~~0.02 &                 -- &                 -- && ~~~0.0228 (0.0103) & ~~~0.0220 (0.0082) \\
$\alpha_1$ &  ~~~0.06 &                 -- &                 -- && ~~~0.0613 (0.0169) & ~~~0.0602 (0.0135) \\
 $\beta_1$ &  ~~~0.90 &                 -- &                 -- && ~~~0.8923 (0.0328) & ~~~0.8952 (0.0259) \\\hline
\end{tabular}
\end{center}
{\it Note:} For each cell, the statistics given are based on 500 simulated samples. The mean and root mean squared error (in parentheses) for each estimator are shown.
\end{table}

We consider a log-Gamma-GARMA-GARCH model with lag orders $p=q=r=s=1$. Specifically,
\[
    y_t\sim \hbox{Gam}(c_t,c_t/\eta_t),\quad \log y_t=\phi_0+\phi_1\log y_{t-1}+\varepsilon_t+\delta_1\varepsilon_{t-1},\quad \sigma^2_t=\omega+\alpha_1\varepsilon^2_{t-1}+\beta_1\sigma^2_{t-1},
\]
with $\varepsilon_t=\log y_{t}-\log\eta_t-\psi(c_t)+\log c_t$ and $\sigma_t^2=\psi_1(c_t)$. The specific true parameter values are assigned as $\phi_0=0$, $\phi_1=0.95$, $\delta_1=-0.65$, $\omega=0.02$, $\alpha_1=0.06$, and $\beta_1=0.90$. We consider different sample sizes $T=100$, $500$ or $2000$. Each simulation is repeated for 500 times. The mean and standard errors of the estimates are presented in Table \ref{table1}.

Several observations can be {drawn} from Table \ref{table1}. First, under the true GARMA-GARCH model, both GMLE and MLE perform well, especially when the sample size is large. This result is  consistent with the theoretical result of \cite{LH1994} and \cite{Escanciano2009} that the GMLE is consistent under relatively weak conditions.
GMLE performs slightly worse than the MLE when sample size is large, with slightly larger biases and root mean squared errors (in parentheses) than that of the MLE. This is possibly due to the loss of efficiency when using GMLE.
Nevertheless, the GMLE can serve as good starting values for ML estimation.
Second, under the mis-specified log-Gamma-M-GARMA model,
both GMLE and MLE produce relative accurate estimates for the ARMA parameters, with slightly larger variances than those under the true model.
The parameter $c$ under log-Gamma-M-GARMA model is the extra time invariant shape parameter in the conditional Gamma distribution, as mentioned in Section 2.2.

\subsection{Logit-Beta-GARMA(1,1)-GARCH(1,1) model}

\begin{table}[!b]
\renewcommand{\arraystretch}{1.00}
\small \tabcolsep 3.7mm
\caption{Simulation results of the logit-Beta-GARMA-GARCH model}\label{table2}
\begin{center}
\begin{tabular}{ccccccccccccc}
\hline\hline
           &          &  \multicolumn{2}{c}{Logit-Beta-M-GARMA} &&\multicolumn{2}{c}{Logit-Beta-GARMA-GARCH}\\ \cline{3-4}\cline{6-7}
 Parameter &     True &              GMLE  &               MLE  &&              GMLE  &               MLE  \\ \hline
           &                                  \multicolumn{6}{c}{$T=100$}                                \\
  $\phi_0$ &  $-$0.10 & $-$0.1790 (0.1466) & $-$0.1758 (0.1397) && $-$0.1531 (0.1080) & $-$0.1531 (0.1071) \\
  $\phi_1$ &  ~~~0.90 & ~~~0.8203 (0.1473) & ~~~0.8237 (0.1403) && ~~~0.8459 (0.1090) & ~~~0.8460 (0.1079) \\
$\delta_1$ &  $-$0.50 & $-$0.4363 (0.2216) & $-$0.4380 (0.2178) && $-$0.4601 (0.1572) & $-$0.4586 (0.1578) \\
    $\tau$ &          & ~~~90.418 (42.408) & ~~~90.476 (42.386) &&                 -- &                 -- \\
  $\omega$ &  ~~~0.01 &                 -- &                 -- && ~~~0.0135 (0.0092) & ~~~0.0136 (0.0092) \\
$\alpha_1$ &  ~~~0.45 &                 -- &                 -- && ~~~0.3940 (0.1796) & ~~~0.3940 (0.1791) \\
 $\beta_1$ &  ~~~0.45 &                 -- &                 -- && ~~~0.3999 (0.2311) & ~~~0.3964 (0.2306) \\\hline
           &                                  \multicolumn{6}{c}{$T=500$}                                \\
  $\phi_0$ &  $-$0.10 & $-$0.1154 (0.0497) & $-$0.1142 (0.0524) && $-$0.1085 (0.0261) & $-$0.1084 (0.0261) \\
  $\phi_1$ &  ~~~0.90 & ~~~0.8844 (0.0498) & ~~~0.8856 (0.0524) && ~~~0.8914 (0.0262) & ~~~0.8915 (0.0262) \\
$\delta_1$ &  $-$0.50 & $-$0.4922 (0.1042) & $-$0.4906 (0.4906) && $-$0.4908 (0.0545) & $-$0.4906 (0.0542) \\
    $\tau$ &          & ~~~69.287 (19.761) & ~~~69.375 (19.658) &&                 -- &                 -- \\
  $\omega$ &  ~~~0.01 &                 -- &                 -- && ~~~0.0106 (0.0032) & ~~~0.0106 (0.0033) \\
$\alpha_1$ &  ~~~0.45 &                 -- &                 -- && ~~~0.4330 (0.0798) & ~~~0.4341 (0.0801) \\
 $\beta_1$ &  ~~~0.45 &                 -- &                 -- && ~~~0.4450 (0.0836) & ~~~0.4440 (0.0839) \\\hline
           &                                  \multicolumn{6}{c}{$T=2000$}                                \\
  $\phi_0$ &  $-$0.10 & $-$0.1068 (0.0280) & $-$0.1047 (0.0240) && $-$0.1017 (0.0111) & $-$0.1017 (0.0110) \\
  $\phi_1$ &  ~~~0.90 & ~~~0.8932 (0.0281) & ~~~0.8953 (0.0238) && ~~~0.8983 (0.0109) & ~~~0.8984 (0.0108) \\
$\delta_1$ &  $-$0.50 & $-$0.4938 (0.0648) & $-$0.4948 (0.0563) && $-$0.4983 (0.0254) & $-$0.4981 (0.0251) \\
    $\tau$ &          & ~~~64.576 (11.430) & ~~~64.612 (11.441) &&                 -- &                 -- \\
  $\omega$ &  ~~~0.01 &                 -- &                 -- && ~~~0.0101 (0.0014) & ~~~0.0101 (0.0014) \\
$\alpha_1$ &  ~~~0.45 &                 -- &                 -- && ~~~0.4478 (0.0434) & ~~~0.4480 (0.4480) \\
 $\beta_1$ &  ~~~0.45 &                 -- &                 -- && ~~~0.4477 (0.0388) & ~~~0.4476 (0.0387) \\\hline
\end{tabular}
\end{center}
{\it Note:} For each cell, the statistics given are based on 500 simulated samples. The mean and root mean squared error (in parentheses) for each estimator are shown.
\end{table}

In this example, we simulate the time series of proportions from the following logit-Beta-GARMA(1,1)-GARCH(1,1) model
\[
    y_t\sim \hbox{Beta}(a_t,b_t),\quad \logit(y_t)=\phi_0+\phi_1\logit(y_{t-1})+\varepsilon_t+\delta_1\varepsilon_{t-1},\quad
    \sigma^2_t=\omega+\alpha_1\varepsilon^2_{t-1}+\beta_1\sigma^2_{t-1},
\]
where $\varepsilon_t=\logit(y_{t})-\psi(a_t)+\psi(b_t)$ and $\sigma_t^2=\psi_1(a_t)+\psi_1(b_t)$. Again we consider three sample sizes $T=100$, $T=500$ and $T=2000$. The true parameter values are set as $\phi_0=-0.10$, $\phi_1=0.90$, $\delta_1=-0.50$, $\omega=0.01$, $\alpha_1=0.45$, and $\beta_1=0.45$. We carry out 500 repeated experiments for each simulation.

Table \ref{table2} reports the empirical performance of the two estimators.
It is seen that, under the true model, the GMLE and MLE perform very similarly in both small and large samples. This implies that the GMLE can be used as a surrogate of the MLE for this model. On the other hand, when the model is mis-specified as the logit-Beta-M-GARMA model, the corresponding GMLE and MLE of the ARMA parameters
%seems still quite well for estimating the ARMA parameters although they
have larger biases and root mean squared errors (in parentheses) than the GMLE and MLE under the true logit-Beta-GARMA-GARCH model. The parameter $\tau$ under the logit-Beta-M-GARMA model is the extra time invariant parameter in the conditional Beta distribution.

\section{Applications}

\subsection{High-frequency realized volatility}\label   {empiricalapp1}

\begin{table}[!b]
\renewcommand{\arraystretch}{1.20}
\small \tabcolsep 5.0mm
\caption{Estimation results of the log-Gamma-GARMA-GARCH model}\label{table3}
\begin{center}
\begin{tabular}{ccccccccccccc}
\hline\hline
           &   \multicolumn{2}{c}{Log-Gamma-M-GARMA}   && \multicolumn{2}{c}{Log-Gamma-GARMA-GARCH} \\ \cline{2-3}\cline{5-6}
 Parameter &                GMLE &                 MLE &&                GMLE &                 MLE \\ \hline
  $\phi_0$ &  $-$0.0582 (0.0198) &  $-$0.0609 (0.0221) &&  $-$0.0623 (0.0179) &  $-$0.0816 (0.0228) \\
  $\phi_1$ &  ~~~0.9552 (0.0115) &  ~~~0.9532 (0.0120) &&  ~~~0.9533 (0.0108) &  ~~~0.9422 (0.0126) \\
$\delta_1$ &  $-$0.4067 (0.0386) &  $-$0.3517 (0.0371) &&  $-$0.4315 (0.0394) &  $-$0.3829 (0.0424) \\
       $c$ &  ~~~2.4924 (0.1013) &  ~~~2.4992 (0.1126) &&                 --  &                 --  \\
%  $\sigma$ &  ~~~0.6237 (0.0149) &                 --  &&                 --  &                 --  \\
  $\omega$ &                 --  &                 --  &&  ~~~0.0433 (0.0228) &  ~~~0.0534 (0.0237) \\
$\alpha_1$ &                 --  &                 --  &&  ~~~0.0583 (0.0224) &  ~~~0.1093 (0.0362) \\
 $\beta_1$ &                 --  &                 --  &&  ~~~0.8311 (0.0718) &  ~~~0.8053 (0.0651) \\ \hline
    Loglik &              258.79 &              260.11 &&              191.63 &              272.59 \\
    AIC    &           $-$0.5303 &	         $-$0.5330 &&           $-$0.3863 &           $-$0.5548 \\
    BIC    &           $-$0.5100 &	         $-$0.5127 &&           $-$0.3559 &           $-$0.5244 \\
    RSS    &              3034.3 &              3026.4 &&              3029.5 &              2919.7 \\
   JB-test &        30.58$^{**}$ &        26.31$^{**}$ &&        34.78$^{**}$ &        30.00$^{**}$ \\
   $Q$(1)  &              0.7653 &              0.3027 &&              1.1873 &              0.0007 \\
   $Q$(5)  &              4.6380 &              6.2492 &&              6.7898 &              4.9045 \\
   $Q$(22) &              28.301 &              29.504 &&              26.229 &              28.331 \\
  $Q^2$(1) &        4.148$^{**}$ &              2.8142 &&              0.2282 &              0.2198 \\
  $Q^2$(5) &        17.84$^{**}$ &        14.359$^{*}$ &&              2.9546 &              2.9060 \\
 $Q^2$(22) &        52.18$^{**}$ &        50.35$^{**}$ &&              9.5137 &              13.797 \\ \hline
\end{tabular}
\end{center}
{\it Note:} ``$**$'' and ``$*$'' indicate that the test statistic is significant at 1\% and 5\% levels, respectively. The standard deviation
errors of parameter estimates are reported in parentheses.
\end{table}

In this application, we employ the proposed log-Gamma-GARMA-GARCH model \eqref{gammadist} to capture the time-varying volatility phenomenon of the realized volatility. Realized volatility has been extensively modeled and studied in financial econometrics, see for example \cite{BNS2002}, \cite{HL2005}, \cite{TOW2009}, and \cite{ZS2014}. In particular, a recent study related to this paper by \cite{CMPP2008} showed that allowing for time-varying volatility of the realized volatility and logarithmic realized variance substantially improves model fitting as well as predictive performance. We use the 5-min daily realized volatility (RV5m, computed by the sum of squared 5-minute log returns) of Standard \& Poor 500 Index (SP500), taken from the ``Oxford-Man Institute's realized library'' (version 0.3, available at the website: {\it http://realized.oxford-man.ox.ac.uk}). The data is sampled from January 3, 2017 to June 30, 2020 with $873$ observations. Based on the extended autocorrelation function \citep{TT1984,CMC2013}, the order is selected as $p=1$ and $q=1$ for the log-Gamma-M-GARMA process. We then set $p=q=r=s=1$ for the log-Gamma-GARMA-GARCH model.

Table \ref{table3} shows the estimation performance of different estimators, including the GMLE and MLE of the log-Gamma-M-GARMA and log-Gamma-GARMA-GARCH models, respectively. Again, the shape parameter $c$ is the fixed extra parameter for the log-Gamma-M-GARMA model.
In the top panel of Table~\ref{table1} we report the parameter estimates and their standard errors. In the bottom panel we report a few statistics for model validation. The first three rows provide the maximum log likelihood, AIC and BIC, respectively. In the fourth row, RSS stands for root residual sum of squares defined by $\hbox{RSS} =\sum_{t=1}^T(y_t-\hat\mu_t)^2$. The fifth row reports the test statistic of \cite{JB1987} for normality. The quantity $Q(m)$ denotes the Box-Ljung test statistic with $m$ lags \citep{LB1978}. The statistic $Q^2(m)$ is the portmanteau test statistic based on squared standardized residuals $\hat e_t^2$, which are defined as $\hat e_t^2=\hat\varepsilon_t^2/\hat\sigma_t^2$, where $\hat\sigma_t^2=\psi_1(\hat{c}_t)$. This statistic is used to test whether the conditional heteroskedasticity has been accommodated by the model \citep{ML1983}.

\begin{figure}[!b]
\begin{center}
    \includegraphics[width=1\textwidth]{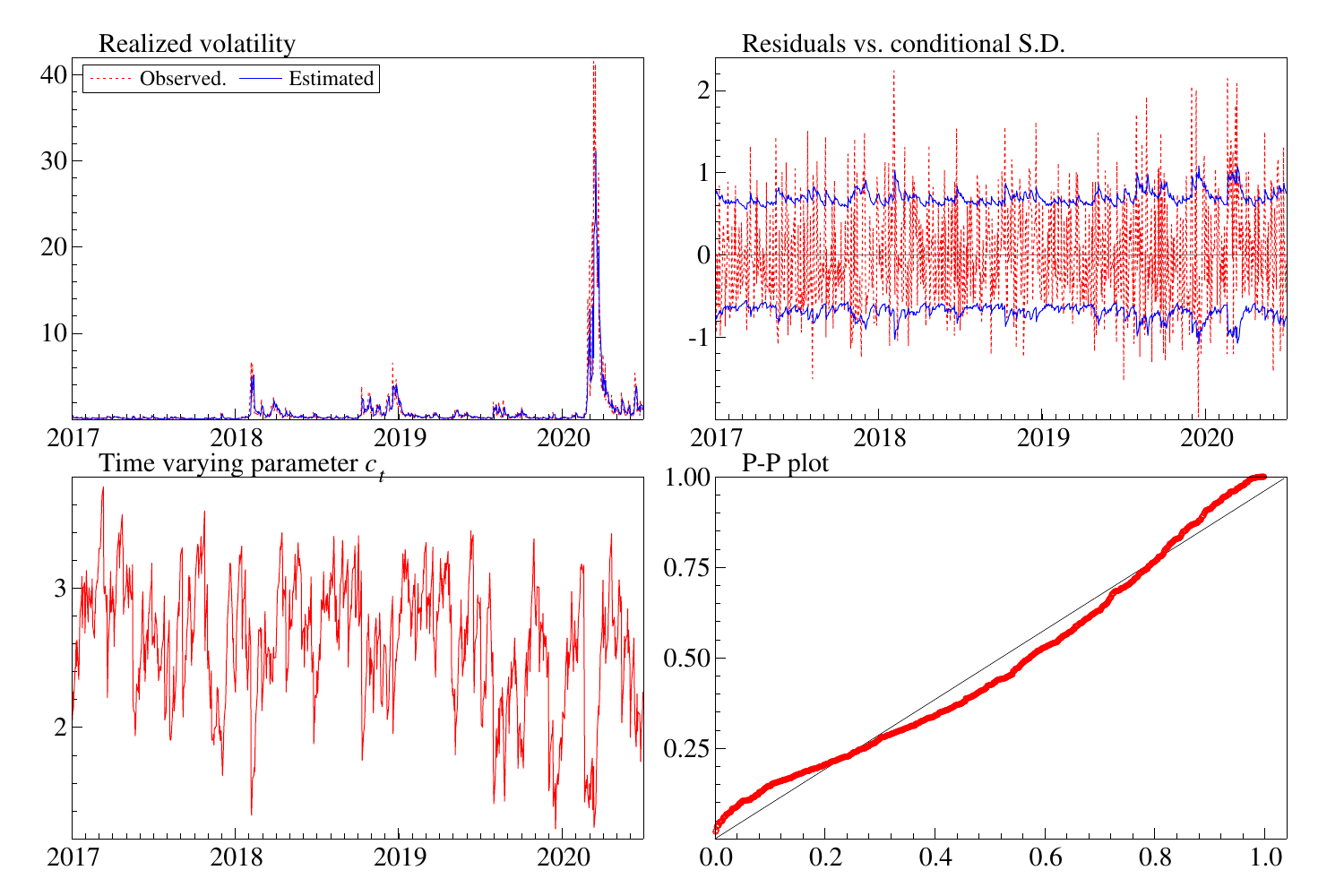}
\caption{Modeling S\&P 500 realized volatility using GARMA-GARCH model.}\label{fig:log-Gamma-GARMA-GARCH}
\end{center}
\end{figure}

The parameter estimates from the GMLE and MLE seem to be close under both models, indicating that the GMLE is a good estimator. However, comparing the left and right panels of Table \ref{table3}, we see that including the GARCH process is indeed appropriate since the coefficients $\alpha_1$ and $\beta_1$ are significant, and the resulting AIC and BIC values are smaller.

Based on the statistics $Q(m)$ and $Q^2(m)$ for various $m$, it can be seen that the log-Gamma-GARMA-GARCH model is more suitable for the data than the log-Gamma-M-GARMA model
since it captures the conditional heteroskedasticity in the process adequately, while the latter one fails to do so. Moreover, the values of log-likelihood function and RSS show that including the GARCH process improves the in-sample fitting performance greatly.
In addition, the normality test shows that the residuals from the GMLE do not follow a normal distribution, indicating that the Gaussian assumption for the innovations is not suitable here.

Figure \ref{fig:log-Gamma-GARMA-GARCH} shows some features of the estimated log-Gamma-GARMA(1,1)-GARCH(1,1) model. The upper-left panel is a plot of the original time series $y_t$ and the fitted values $\hat{y}_t$, showing a good fit to the data. The upper-right panel compares the estimated residuals $\hat\varepsilon_t$ with conditional standard deviation $\hat\sigma_t$, further showing a good description of the conditional heteroskedasticity in the process. The lower left panel presents the estimated time-varying parameter $\hat{c}_t=\psi^{-1}_1(\hat{\sigma}_t^2)$, indicating that this shape parameter varies between 1 and 4, mostly around 2 and 3. We also construct a P-P plot to check the conditional distribution assumption of the model. Specifically, if the distribution assumption is valid, then $\nu_t=F(y_t\mid \gamma_t)$ should follow the uniform distribution, where $F$ is the cumulative distribution function. The lower right panel plots the quantiles of $\hat{\nu}_t=F(y_t\mid \hat{\gamma}_t)$ versus the quantiles of the uniform distribution. It suggests that the conditional Gamma assumption is reasonable.

\subsection{U.S. personal saving rate}

\begin{table}[!t]
\renewcommand{\arraystretch}{1.20}
\small \tabcolsep 5.0mm
\caption{Estimation results of the logit-Beta-GARMA-GARCH model}\label{table4}
\begin{center}
\begin{tabular}{ccccccccccccc}
\hline\hline
           &   \multicolumn{2}{c}{Logit-Beta-M-GARMA}  && \multicolumn{2}{c}{Logit-Beta-GARMA-GARCH}\\ \cline{2-3}\cline{5-6}
 Parameter &                GMLE &                 MLE &&                GMLE &                 MLE \\ \hline
  $\phi_0$ &  $-$0.0140 (0.0119) &  $-$0.0123 (0.0116) &&  ~~~0.0049 (0.0078) &  ~~~0.0047 (0.0079) \\
  $\phi_1$ &  ~~~0.6148 (0.0521) &  ~~~0.6275 (0.0527) &&  ~~~0.6027 (0.0590) &  ~~~0.5979 (0.0586) \\
  $\phi_2$ &     0.0775 (0.0610) &  ~~~0.0721 (0.0597) &&  ~~~0.1956 (0.0705) &  ~~~0.2019 (0.0718) \\
  $\phi_3$ &  $-$0.0191 (0.0611) &  $-$0.0260 (0.0607) &&  $-$0.0031 (0.0601) &  $-$0.0054 (0.0612) \\
  $\phi_4$ &  ~~~0.1362 (0.0610) &  ~~~0.1432 (0.0607) &&  ~~~0.1163 (0.0484) &  ~~~0.1135 (0.0488) \\
  $\phi_5$ &  ~~~0.1460 (0.0521) &  ~~~0.1425 (0.0458) &&  ~~~0.0567 (0.0381) &  ~~~0.0547 (0.0379) \\
    $\tau$ &  ~~~121.56 (10.486) &  ~~~121.59 (5.0038) &&                 --  &                 --  \\
  $\omega$ &                 --  &                 --  &&  ~~~0.0063 (0.0012) &  ~~~0.0065 (0.0013) \\
$\alpha_1$ &                 --  &                 --  &&  ~~~0.6964 (0.0777) &  ~~~0.7058 (0.0767) \\
 $\beta_1$ &                 --  &                 --  &&  ~~~0.2495 (0.0639) &  ~~~0.2397 (0.0623) \\ \hline
    Loglik &              616.34 &              616.39 &&              689.64 &              689.76 \\
    AIC    &           $-$3.3852 &	         $-$3.3855 &&           $-$3.7813 &           $-$3.7820 \\
    BIC    &           $-$3.3097 &	         $-$3.3099 &&           $-$3.6842 &           $-$3.6848 \\
    RSS    &              0.6798 &              0.6796 &&              0.6962 &              0.6977 \\
   JB-test &       534.27$^{**}$ &       550.41$^{**}$ &&       563.58$^{**}$ &        556.53$^{**}$ \\
    $Q$(1) &              3.3e-6 &              0.0542 &&              3.2231 &              3.4302 \\
    $Q$(3) &              1.6662 &              1.6532 &&              7.6838 &              7.9635 \\
   $Q$(12) &              9.7379 &              9.8820 &&              19.555 &              19.805 \\
  $Q^2$(1) &        41.224$^{**}$ &      40.405$^{**}$ &&              0.0344 &              0.0454 \\
  $Q^2$(3) &        44.417$^{**}$ &      43.264$^{**}$ &&              2.8419 &              3.1277 \\
 $Q^2$(12) &        46.474$^{**}$ &      45.255$^{**}$ &&              15.216 &              17.064 \\ \hline
\end{tabular}
\end{center}
{\it Note:} ``$**$'' and ``$*$'' indicate that the test statistic is significant at 1\% and 5\% levels, respectively. The standard deviation
errors of parameter estimates are reported in parentheses.
\end{table}

In this example we study the monthly U.S. personal saving rate from January 1990 to December 2019 with $360$ observations, shown in the upper left panel of Figure \ref{fig:logit-Beta-GARMA-GARCH}. This seasonal adjusted monthly series is retrieved from FRED, the Federal Reserve Bank of St. Louis (https://fred.stlouisfed.org/series/PSAVERT). We assume that the saving rate is in the range $(0,0.15)$, and hence multiply the series by $20/3$. The unit-root test with Phillips-Perron test statistic indicates that the series is stationary. Using BIC and standard ARMA modeling, the order of the ARMA process is determined as $p=5$ and $q=0$. We then use $p=5$, $q=0$ and $r=s=1$ for the logit-Beta-GARMA-GARCH model, which can be represented as
\begin{align*}
  y_t\mid\F_{t-1}\sim \hbox{Beta}(a_t,b_t),&\quad \logit(y_t)=\phi_0+\sum_{j=1}^5\phi_j\logit(y_{t-j})+\varepsilon_t,\\
  \sigma_t^2=&\omega+\alpha_1\varepsilon_{t-1}^2+\beta_1\sigma_{t-1}^2,
\end{align*}
where $\varepsilon_t=\logit(y_t)-g(a_t,b_t)$. We then estimate the logit-Beta-M-GARMA(5,0) model and logit-Beta-GARMA(5,0)-GARCH(1,1) model using different estimators.

\begin{figure}[!t]
\begin{center}
    \includegraphics[width=1\textwidth]{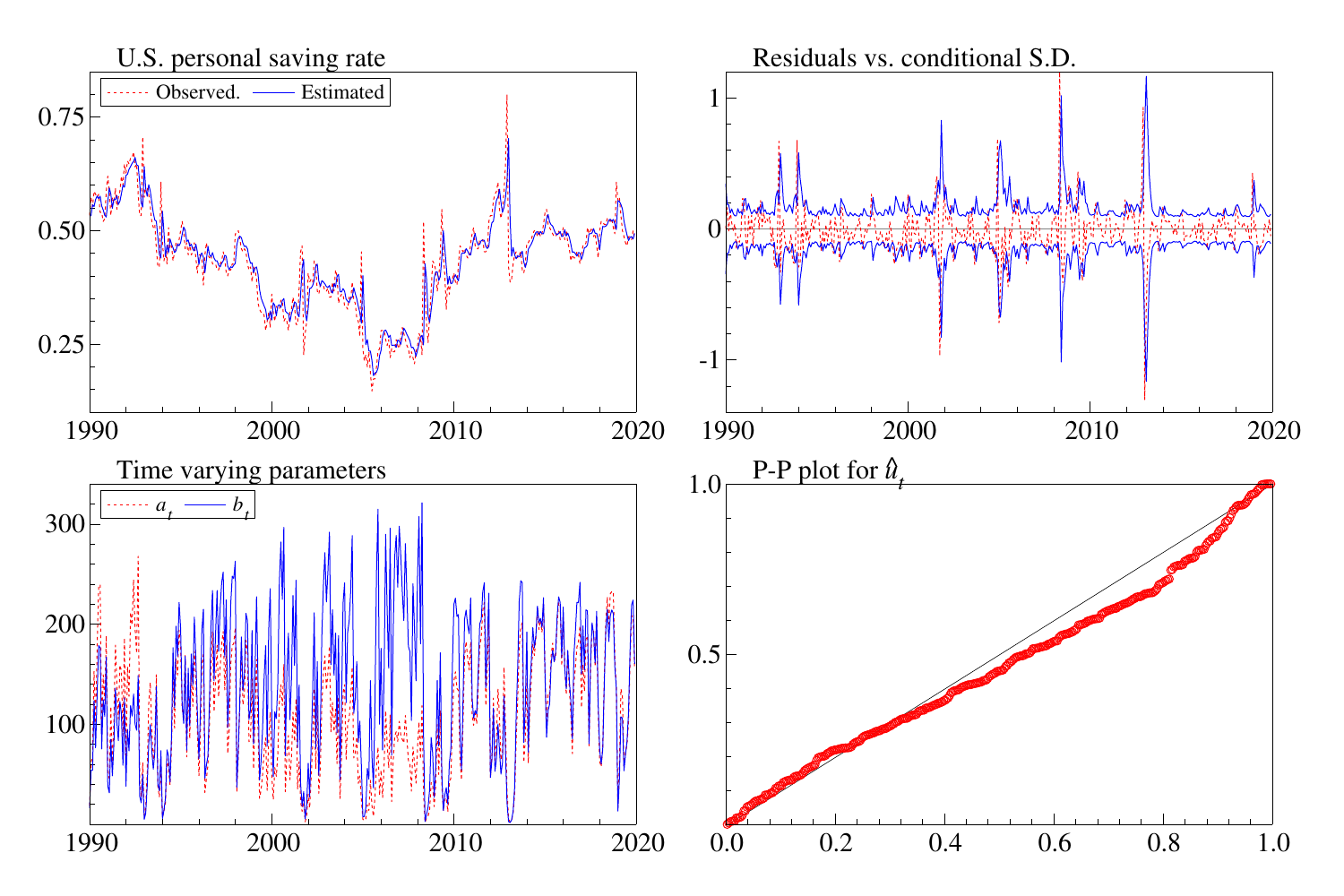}
\caption{Modeling U.S. personal saving rate using GARMA-GARCH model.}\label{fig:logit-Beta-GARMA-GARCH}
\end{center}
\end{figure}

Table \ref{table4} reports the estimation results and some summary statistics.
Again, the log-Gamma-M-GARMA model has an extra time invariant parameter $\tau$.
From the parameter estimates shown in the upper panel, we can see that
the estimated AR coefficients are significantly different under the two models. In particular,
the coefficient $\phi_5$ is statistically significant under the logit-Beta-M-GARMA model, but insignificant under
the logit-Beta-GARMA-GARCH model.
Moreover, the GARCH coefficients $\alpha_1$ and $\beta_1$ are strongly significant, indicating that adding the GARCH process to the logit-Beta-M-GARMA model is reasonable.

In the lower-left panel of Table \ref{table4}, the significant $Q^2(m)$ statistics show that the squared standardized residuals of the logit-Beta-M-GARMA model, $\hat{e}_t^2=\hat{\varepsilon}_t^2/\hat{\sigma}_t^2$, have strong autocorrelations.
However, under the logit-Beta-GARMA-GARCH model, the portmanteau test results in the lower-right panel show no significant autocorrelations in the squared standardized residuals $\hat{e}_t^2=\hat{\varepsilon}_t^2/\hat{\sigma}_t^2$,
indicating the advantage of including the GARCH type conditional variance process in the model. %In addition, the MLE for the logit-Beta-GARMA-GARCH model performs closely to the GMLE according to the values of

Figure \ref{fig:logit-Beta-GARMA-GARCH} depicts some features of the estimated
logit-Beta-GARMA(5,0)-GARCH(1,1) model using the MLE. The estimated conditional mean $\hat{y}_t=\hat{a}_t/(\hat{a}_t+\hat{b}_t)$ is shown together with the observed series in the upper left panel of the figure.  The conditional standard deviation $\hat{\sigma}_t$ (upper right panel) are shown in the upper right panel along with the estimated residuals $y_t-\hat{y}_t$, showing reasonable coverage. The lower left panel shows the estimated time-varying parameters $\hat{a}_t$ and $\hat{b}_t$, indicating strong time varying behavior of the two parameters.
The lower right panel shows the
P-P plot similar to that in Figure \ref{fig:log-Gamma-GARMA-GARCH},
indicating that the logit-Beta-GARMA-GARCH model is reasonable for this time series data.

\subsection{Stock returns}

In this example, the proposed GHSST-GARMA-GARCH model in Section 2.3 is employed to investigate the skewness, fat-tail and volatility in stock returns. We use the daily log returns of Standard \& Poor 500 Index (SP500) from January 3, 2000 to June 30, 2020 with $5157$ observations, obtained from Yahoo Finance ({\it https://finance.yahoo.com/}). The upper-left panel of Figure \ref{fig:GHSST-GARMA-GARCH} shows the time series.

\begin{table}[!t]
\renewcommand{\arraystretch}{1.20}
\small \tabcolsep 10.4mm
\caption{Estimation results of GHSST-GAR(1)-GARCH(1,1) model.}\label{table5}
\begin{center}
\begin{tabular}{cccccccccccc}
\hline\hline
 Parameter &                  Normal  &             Student-$t$  &             GHSST        \\ \hline
  $\phi_0$ &       ~~~0.0628 (0.0110) &       ~~~0.0764 (0.0102) &       ~~~0.0636 (0.0104) \\
  $\phi_1$ &       $-$0.0602 (0.0151) &       $-$0.0632 (0.0140) &       $-$0.0729 (0.0141) \\
  $\omega$ &       ~~~0.0203 (0.0029) &       ~~~0.0088 (0.0019) &       ~~~0.0121 (0.0025) \\
$\alpha_1$ &       ~~~0.1275 (0.0100) &       ~~~0.0807 (0.0078) &       ~~~0.1196 (0.0102) \\
 $\beta_1$ &       ~~~0.8627 (0.0099) &       ~~~0.8771 (0.0105) &       ~~~0.8773 (0.0103) \\
     $\nu$ &                      --  &       ~~~5.9138 (0.5056) &       ~~~6.9631 (0.6246) \\
    $\tau$ &                      --  &                      --  &       $-$0.2006 (0.0502) \\ \hline
    Loglik &               $-$7089.7  &               $-$6969.3  &               $-$6951.1  \\
    AIC    &                  2.7515  &	        	     2.7052  &                  2.6985  \\
    BIC    &                  2.7611  &                  2.7145  &                   2.7107 \\
%    BIC    &                  2.7562  &                  2.7111  &                  2.7057  \\
%    RSS   &                  5383.15  &                 5381.17  &                 5376.37  \\
%   JB-test &             1849$^{**}$  &             1821$^{**}$  &             1799$^{**}$  \\
    $Q$(1) &                  0.5559  &                  0.8591  &                  2.3289  \\
    $Q$(5) &                  8.0251  &                  8.4076  &                  10.189  \\
   $Q$(22) &                  32.868  &                  33.618  &                  30.475  \\
  $Q^2$(1) &                  1.1042  &                  0.6035  &                  0.3015  \\
  $Q^2$(5) &                  5.1283  &                  3.8969  &                  4.4023  \\
 $Q^2$(22) &                  21.172  &                  17.881  &                  23.367  \\ \hline
\end{tabular}
\end{center}
{\it Note:} ``$**$'' and ``$*$'' indicate that the test statistic is significant at 1\% and 5\% levels, respectively. The standard deviation
errors of parameter estimates are reported in parentheses.
\end{table}

We consider a GHSST-GAR(1)-GARCH(1,1) model specified as follows:
\begin{align*}
    y_t\mid\F_{t-1}\sim&~\mathrm{GHSST}(\xi_t,\varsigma_t,\nu,\tau),\quad y_t=\phi_0+\phi_1y_{t-1}+\varepsilon_t,\\
    &\sigma_t^2=\omega+\alpha_1\varepsilon_{t-1}^2+\beta_1\sigma_{t-1}^2.
\end{align*}
We also estimate the AR(1)-GARCH(1,1) models with normal and Student-$t$ distributions for comparison, which can be seen as special cases of the GHSST distribution.

Table \ref{table5} reports the estimation results based on different models. First, according to the the portmanteau tests in the lower panel, the residuals and squared standardized residuals have no significant autocorrelations for all three models, indicating that all the models can explain the conditional heteroskedasticity in the process.
Second, the parameter estimates $\hat{\nu}$ and $\hat{\tau}$ for  GHSST-GAR(1)-GARCH(1,1) model in the upper panel are significant, showing that the log returns are indeed leptokurtic and left skewed. We also see the advantage of the GHSST model over the normal and Student-$t$ settings, with its larger log-likelihood value.

\begin{figure}[!t]
\begin{center}
    \includegraphics[width=1\textwidth]{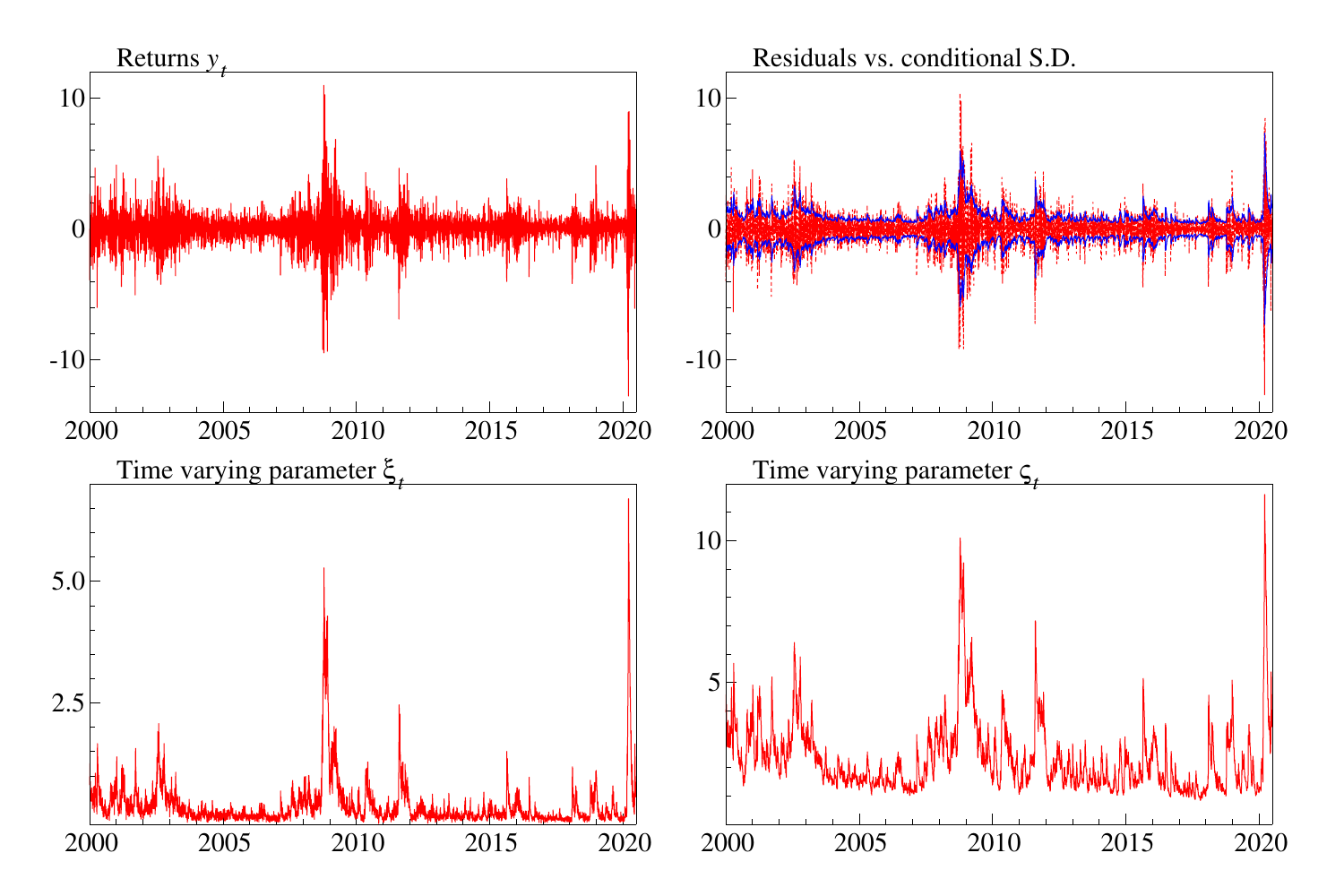}
\caption{Modeling S\&P 500 returns using GHSST-GAR(1)-GARCH(1,1) model.}\label{fig:GHSST-GARMA-GARCH}
\end{center}
\end{figure}

Figure \ref{fig:GHSST-GARMA-GARCH}
plots the residuals along with conditional standard deviation $\hat{\sigma}_t$,
(the upper right panel) and the two time varying parameters $\hat{\xi}_t$ and $\hat{\varsigma}_t$ in the lower left and right panels, respectively. The figures show that the GHSST-GAR(1)-GARCH(1,1) captures the salient features of the observed time series.

\section{Concluding remarks}

The GARMA model is a useful class of data-based models for analyzing non-Gaussian time series. Since it only models the conditional mean process though an ARMA formation, it lacks the ability to address conditional heteroskedasticity often encountered in applications. In this paper we extend the GARMA model to GARMA-GARCH model with an additional model assumption on the conditional variance process, under a GARCH formation. In addition, three special GARMA-GARCH models are proposed for non-negative time series, proportional time series, and skewed and heavy tailed financial time series.
Furthermore, maximum likelihood and quasi Gaussian likelihood estimation procedures are introduced, and their finite sample performances are illustrated. We find that the GMLE performs very well and the corresponding parameter estimates can be used as starting values of the MLE procedure. Three real data examples are used to demonstrate the properties of the proposed models.

There are several directions to extend the model. It is possible to introduce additional model assumptions on the higher moments of the conditional distribution so to capture additional time varying features of the conditional distribution. Modeling additional higher moments can be fruitful even when the conditional distribution is specified with a small number of parameters, since a generalized moment method type of procedure can be used to enhance the parameter estimation. Another direction is to extend the model to analyzing multivariate non-Gaussian time series, a topic less studied. The GARMA-GARCH framework conveniently introduces dependencies among the component time series though dependencies among the conditional mean and variance processes under a multivariate version of the ARMA and GARCH formations.

%\section*{Data availability statement}

%All data used for this project can be obtained from the corresponding author via email.

\section*{Acknowledgement}

\thanks{Zheng acknowledges financial supports from  the National Science Foundation of China (No. 71973110, 71371160) and the Fundamental Research Fund  for the Central Universities (No. 20720191038).}

%%%%%%%%%%%%%%%%%%%%%%%%%%%%%%%%%%%%%%%%%%%%%%%%%%%%%%%%%%%%%%%%%%%%
%%                          References
%%%%%%%%%%%%%%%%%%%%%%%%%%%%%%%%%%%%%%%%%%%%%%%%%%%%%%%%%%%%%%%%%%%%

\renewcommand{\thesection}{\Alph{section}}
\setcounter{section}{0}

\newpage
\section{Appendix}

We give the proofs of Theorem~\ref{thm:gamma} and \ref{thm:beta} in the appendix. The proof relies on the geometric drift condition discussed in Chapter~15 of \cite{meyn:2009}. Lemma~1 of \cite{ZXC2015} is a convenient wrap-up of the tool that we need, so we repeat it here.

Let $\{X_n\}_{n\geq 0}$ be a Markov chain on the state space
$\vX$, equipped with some $\sigma$-field $\mathcal{B}(\vX)$. Let
$\{P(x,A),\;x\in\vX,A\in\sB(\vX)\}$ be the transition probability
kernel. The {\it geometric drift condition} requires that there exists an extended-valued
function $\mathcal{V}:\;\vX\rightarrow[1,\infty]$, a measurable set $C$, and
constants $b<\infty$, $\beta>0$ such that
\begin{equation}
  \label{eq:V4}%\tag{D}
  \Delta {\cal V}(x):=\int P(x,dy){\cal V}(y) - {\cal V}(x)
  \leq -\beta {\cal V}(x)+bI_C(x), \quad x\in\vX.
\end{equation}

\begin{lemma}
  \label{thm:regular}
  Suppose $\{X_t\}$ is a $\psi$-irreducible and aperiodic Markov chain. If for some $m$,
  the skeleton $\{X_{mt}\}$ satisfies the drift condition \eqref{eq:V4}
  for a petite set $C$
  and a function ${\cal V}$ which is everywhere finite. Then $\{X_t\}$ is
  geometrically ergodic, and $\int_{\vX} V(x)\pi(dx)<\infty$, where
  $\pi$ is the unique invariant probability measure.
\end{lemma}

\begin{proof}[Proof of Theorem~\ref{thm:gamma}]

According to Remark~\ref{rk:gamma0}, we can embed the process $\{\varepsilon_t\}$ into the following Markov chain $\{Z_t\}$. Given $Z_{t-1}$, we generate $\tilde y_t$ as $\tilde y_t\sim \mathrm{Gam}(c_t,1)$, where the parameter $c_t$ is determined by
$\psi_1(c_t)=\sigma_t^2:=\sigma^2+\beta'\bm A Z_{t-1}$.
We then set $\varepsilon_t=h(\tilde y_t)-\psi(c_t)$, $\zeta_t=\varepsilon_t^2-\psi_1(c_t)$, and define $Z_t=\hA Z_{t-1} + (1,0,\ldots,0)'\zeta_t$. It holds that
  \begin{align*}
    E(\|Z_t\|^2\mid Z_{t-1})
    & = Z_{t-1}'\hA'\hA Z_{t-1} + \mathrm{Var}(\zeta_t\mid Z_t).
  \end{align*}
  According to \S2.2 of \cite{chan1993},
  \begin{equation*}
    \mathrm{Var}(\zeta_t\mid Z_{t-1}) = \mathrm{Var}(\zeta_t\mid \sigma^2_t) = \psi_3(c_t) - [\psi_1(c_t)]^2,
  \end{equation*}
  where $\psi_m(z)$ is the order-$m$ polygamma function defined as the $(m+1)$-th derivative of the function $\log \Gamma(z)$. The polygamma function has the following two properties
  \begin{equation}
  \label{eq:polygamma}
      \psi_m(z+1) = \psi_m(z) + \frac{(-1)^m m!}{z^{m+1}},\quad
      \psi_m(z) = (-1)^{m+1}m!\sum_{k=0}^{\infty}\frac{1}{(z+k)^{m+1}}.
  \end{equation}
  It follows that for each $\kappa>0$, there exists a $D_\kappa>0$ such that
  \begin{equation}
    \label{eq:1step}
    \mathrm{Var}(\zeta_t\mid \sigma^2_t)  \leq 5(1+\kappa)(\beta'\bm A Z_{t-1})^2+D_\kappa,
  \end{equation}
  and
  \begin{align*}
    E(\|Z_t\|^2\mid Z_{t-1}) & \leq Z_{t-1}'\hA'\hA Z_{t-1} + 5(1+\kappa)(\beta'\bm A Z_{t-1})^2 + D_{\kappa}\\
    & = Z_{t-1}'\hA'\hA Z_{t-1} +  5(1+\kappa) Z_{t-1}'\hA_1'\hA_1 Z_{t-1} + D_{\kappa} \\
    & \leq (1+\kappa)Z_{t-1}'\hB_1 Z_{t-1} + D_\kappa.
  \end{align*}
  Next, by taking a double expectation
  \begin{align*}
    E(\|Z_t\|^2\mid Z_{t-2}) & \leq  (1+\kappa)E[Z_{t-1}' \hB_1 Z_{t-1} \mid Z_{t-2}] + {D_\kappa}.
  \end{align*}
  Applying \eqref{eq:1step} again,
  \begin{align*}
    E[Z_{t-1}' \hB_1 Z_{t-1} \mid Z_{t-2}]
    & = Z_{t-2}'\hA'\hB_1\hA Z_{t-2} + w_1\mathrm{Var}(\zeta_{t-1}\mid \sigma^2_{t-1}) \\
    & \leq Z_{t-2}'\hA'\hB_1\hA Z_{t-2} + 5w_1(1+\kappa)(\beta'\bm A Z_{t-2})^2 + w_1D_\kappa \\
    & = Z_{t-2}'\hA'\hB_1\hA Z_{t-2} + 5(1+\kappa)\lambda Z_{t-2}'\hA_1'\hB_1\hA_1 Z_{t-2}+ w_1D_\kappa\\
    & \leq (1+\kappa) Z_{t-2}'\hB_2 Z_{t-2} + w_1D_\kappa.
  \end{align*}
  It follows that
  \begin{align*}
    E(\|Z_t\|^2\mid Z_{t-2}) & \leq (1+\kappa)^2 Z_{t-2}'\hB_2 Z_{t-2} + (1+\kappa)w_1D_\kappa + D_\kappa,
  \end{align*}
  where $w_1$ is the $(1,1)$-th entry of $\hB_1$.

  Following the same argument recursively, we have for any positive integer $h$,
  \begin{align*}
    E(\|Z_t\|^2\mid Z_{t-h})
    & \leq (1+\kappa)^h Z_{t-h}'\hB_h Z_{t-h} + \sum_{j=0}^{h-1} (1+\kappa)^{j}w_{j}D_\kappa.
  \end{align*}
  where $w_{j}$ is the $(1,1)$-th entry of $\hB_j$ for $j\geq 1$, and
  $w_0=1$.
  Choose $h$, such that the operator norm of $\hB_h$ is less than
  $1-2\iota$ for some $\iota>0$. Choosing $\kappa$ such that
  $(1+\kappa)^h(1-2\iota)<1-\iota$, we have
  \begin{align*}
    E(\|Z_t\|^2\mid Z_{t-h})
    & \leq (1-\iota) \|Z_{t-h}\|^2 + \sum_{j=0}^{h-1} (1+\kappa)^{j}w_{j}D_\kappa.
  \end{align*}
  Since the log-Gamma distribution is absolutely continuous on $\mathbb{R}$, from here it is easy to verify
  that for the skeleton $\{Z_{th}\}_{t\geq 1}$, the drift condition \eqref{eq:V4} is
  met with ${\cal V}(x)=\|x\|^2+1$ for $x\in\mathbb{R}^p$, and other conditions of Lemma~\ref{thm:regular} are also fulfilled. Therefore, the chain $\{Z_t\}$ is geometrically ergodic,
  and $E_\pi \|Z_t\|^2 <\infty$, where $\pi$ is the unique invariant
  probability measure. Because $\varepsilon_t^2=\sigma^2+\beta'Z_t$, it follows that $E_{\pi}\varepsilon_t^4<\infty$.

Finally, once the strictly stationary process $\{\varepsilon_t\}$ has been generated, the condition $\phi(z)\neq 0$ for $|z|\leq 1$ guarantee that the $\{h(y_t)\}$ process generated according to \eqref{eq:arma} admits a strictly stationary solution such that $E[h(y_t)^4]<\infty$.
\end{proof}

\begin{proof}[Proof of Theorem~\ref{thm:beta}]
For the logit-Beta-GARMA-GARCH model, the process $\{\varepsilon_t\}$ cannot be generated by itself, and has to be generated together with $y_t$. Therefore, we need to consider the Markov chain defined in \eqref{eq:MC1}. For notational simplicity, denote $W_t=(X_t',Z_t)'$. %and make the convention that $\diag(\Phi,\hA)$ denotes the block diagonal matrix with diagonal blocks $\Phi$ and $\hA$.
Using the relationship between the Beta and Gamma distribution, it can be shown that
\begin{align*}
    g(a_t,b_t) & =\psi(a_t) - \psi(b_t) = \mu + \delta'\Phi X_{t-1}, \\
    V(a_t,b_t) & = \psi_1(a_t) + \psi_1(b_t) = \sigma^2 + \beta'\hA Z_{t-1},\\
    \Var(\zeta_t\mid W_{t-1}) & = \psi_3(a_t) - [\psi_1(a_t)]^2 + \psi_3(b_t) - [\psi_1(b_t)]^2.
\end{align*}
Suppose for some, the operator norms of both $\hB_h$ and $\Phi^h$ are less than $1-2\iota$ with some $\iota>0$. Let $v_j$ be the $(1,1)$-th entry of $(\Phi^j)'\Phi^j$, and $w_j$ be the $(1,1)$-th entry of $\hB_j$ for $j\geq 1$, and set $v_0=w_0=1$. Note that $w_j>0$ for each $j$.
By \eqref{eq:polygamma}, for each $\kappa>0$, there exists a $D_\kappa>0$ such that
\begin{equation}
    \label{eq:1step_beta}
     v_j\Var(\varepsilon_{t-j}\mid W_{t-j-1}) + w_j\mathrm{Var}(\zeta_{t-j}\mid W_{t-j-1}) \leq 5w_j(1+\kappa)(\beta'\bm A Z_{t-j-1})^2+D_{\kappa},\; \; 0\leq j<h.
\end{equation}
It follows that (using the preceding inequality with $j=0$)
\begin{align*}
    E(\|W_t\|^2\mid W_{t-1}) & \leq X_{t-1}'\Phi'\Phi X_{t-1} + Z_{t-1}'\hA'\hA Z_{t-1} + 5(1+\kappa)(\beta'\bm A Z_{t-1})^2 + D_{\kappa}\\
    & \leq X_{t-1}'\Phi'\Phi X_{t-1} + (1+\kappa)Z_{t-1}'\hB_1 Z_{t-1} + D_\kappa.
\end{align*}
%Let $v_1$ and $w_1$ be the $(1,1)$-th entry of $\Phi'\Phi$ and $\hB_1$ respectively.
Since
\begin{align*}
    E(X_{t-1}'\Phi'\Phi X_{t-1} \mid W_{t-2}) & = X_{t-2}'(\Phi^2)'\Phi^2X_{t-2} + v_1\mathrm{Var}(\varepsilon_{t-1}\mid W_{t-2}),
\end{align*}
and
\begin{align*}
    E(Z_{t-1}'\hB_1 Z_{t-1} \mid W_{t-2}) & = Z_{t-2}'\hA'\hB_1\hA Z_{t-2} + w_1\mathrm{Var}(\zeta_{t-1}\mid W_{t-2}),
\end{align*}
applying \eqref{eq:1step_beta} again with $j=1$ we have
\begin{align*}
    E(\|W_t\|^2\mid W_{t-2}) & \leq X_{t-2}'(\Phi^2)'\Phi^2X_{t-2} + (1+\kappa) Z_{t-2}'\hA'\hB_1\hA Z_{t-2} \\
    & \qquad + 5w_j(1+\kappa)^2(\beta'\bm A Z_{t-j-1})^2+(2+\kappa)D_{\kappa} \\
    & \leq X_{t-2}'(\Phi^2)'\Phi^2X_{t-2} + (1+\kappa)^2Z_{t-2}'\hB_2 Z_{t-2}+(2+\kappa)D_{\kappa}.
\end{align*}
Similar to the proof of Theorem~\ref{thm:gamma}, applying \eqref{eq:1step_beta} recursively, we have
\begin{align*}
    E(\|W_t\|^2\mid W_{t-h})
    & \leq X_{t-h}'(\Phi^h)'\Phi^hX_{t-h} +(1+\kappa)^h Z_{t-h}'\hB_h Z_{t-h} + [(1+\kappa)^h-1]D_\kappa/\kappa.
\end{align*}
Choosing $\kappa$ such that $(1+\kappa)^h(1-2\iota)<1-\iota$, it holds that
\begin{align*}
    E(\|W_t\|^2\mid W_{t-h})
    & \leq (1-\iota) \|W_{t-h}\|^2 + [(1+\kappa)^h-1]D_\kappa/\kappa.
\end{align*}
The rest of the proof follows the same argument of Theorem~\ref{thm:gamma}, so we skip the details. The proof is complete.
\end{proof}

\end{document}